\documentclass[lettersize,journal]{IEEEtran}
\usepackage{amsmath,amsfonts}
\usepackage{cite}
\allowdisplaybreaks[4]
\usepackage[caption=false,font=footnotesize,labelfont=rm,textfont=rm]{subfig}
\usepackage{textcomp}
\usepackage{enumerate}
\usepackage{amssymb}
\usepackage{algorithm}
\usepackage{algpseudocode}
\usepackage{amsthm}
\usepackage{stfloats}
\usepackage{url}
\usepackage{verbatim}
\usepackage{graphicx}
\usepackage{bm}
\usepackage{array,color}
\usepackage{float}
\usepackage{multirow}
\usepackage{makecell}
\graphicspath{{figurepdf/}}
\begin{document}
\hyphenpenalty=100
\newtheorem{proposition}{Proposition}
\newtheorem{lemma}{Lemma}
\title{A Framework of RIS-assisted ICSC User-centric Based Systems: Latency Optimization and Design}
\author{
    Jiahua Wan, Hong Ren,~\IEEEmembership{Member,~IEEE,} Zhiyuan Yu, Zhenkun Zhang, Yang Zhang, 
    
    Cunhua Pan,~\IEEEmembership{Senior Member,~IEEE,} and Jiangzhou Wang,~\IEEEmembership{Fellow,~IEEE}   
\thanks{Part work of this article was accepted by the IEEE International Conference on Acoustics, Speech and Signal Processing Workshops (ICASSPW), Seoul, Korea, April 2024 \cite{77}.}
\thanks{J. Wan, H. Ren, Z. Yu, Z. Zhang, Y. Zhang, and C. Pan are with National Mobile Communications Research Laboratory, Southeast University, Nanjing, China. (email:\{wanjiahua, hren, zyyu, zhenkun\_zhang, y\_zhang, cpan\}@seu.edu.cn). J. Wang is with the School of Engineering, University of Kent, Canterbury CT2 7NT, U.K. (e-mail: j.z.wang@kent.ac.uk).}
} 
\maketitle
\begin{abstract}
    This paper studies a comprehensive framework for reconfigurable intelligent surface (RIS)-assisted integrated communication, sensing, and computation (ICSC) systems with a User-centric focus. The study encompasses two scenarios: the general multi-user equipment (UE) scenario and the simplified single-UE scenario. To satisfy the critical need for time-efficient sensing, we investigate the latency minimization problem, subject to constraints on UEs' transmit power, radar signal-to-interference-plus-noise-ratio (SINR), RIS phase shift, and computation capability. To address the formulated non-convex problem in the multi-UE scenario, we decouple the original problem into two subproblems, where the computational and beamforming settings are optimized alternately. Specifically, for the computational settings, we derive a closed-form solution for the offloading volume and propose a low-complexity algorithm based on the bisection search method to optimize the edge computing resource allocation. Additionally, we employ two equivalent transformations to address the challenge posed by the non-convex sum-of-ratios form in the objective function (OF) of the subproblem related to active and passive beamforming. Several techniques are then combined to address these subproblems. Furthermore, a low-complexity algorithm that offers closed-form solutions is developed for the simplified single UE scenario. Finally, simulation results substantiate the effectiveness of the proposed framework.
\end{abstract}


\section{Introduction}
\IEEEPARstart{W}{ith the advent of} sixth-generation (6G) networks, the explosion in the number of user equipments (UEs) in wireless communication has precipitated a severe scarcity of spectrum resources \cite{69}. Integrated sensing and communication (ISAC) has emerged as a pivotal 6G technology, efficiently fulfilling both communication and sensing functions within the same frequency band, thereby enhancing spectral efficiency \cite{3,4,5}. 

However, the efficient processing of vast distributed ISAC data remains a formidable challenge, necessitating innovative solutions such as the integration of mobile edge computing (MEC) \cite{68} with ISAC. Specifically, MEC facilitates the migration of computation capability to the wireless network edge by deploying edge servers at cellular base stations (BSs), as demonstrated in \cite{62}. This strategic placement enables the execution of high computational tasks in proximity to ISAC application, thereby achieving ultra-low ISAC service latency while minimizing traffic loads. This integration paves the way for integrated communication, sensing, and computation (ICSC), which strategically processes high computational tasks at the network edge, thereby reducing latency and network congestion significantly \cite{29,30,40,41}.

Nonetheless, the efficacy of  ICSC systems can be compromised owing to the potential obstacles in practical scenarios, which may lead to prolonged offloading delays in MEC \cite{24} and diminished received power of sensing signals \cite{16,52}. Fortunately, reconfigurable intelligent surfaces (RISs) can address these challenges by dynamically adjusting the phase of incident electromagnetic waves\cite{72,73,74}, thus enhancing both communication and sensing capabilities in ISAC systems \cite{67,54}. 
Furthermore, RIS is able to reduce delays by establishing additional reliable links and improving channel quality \cite{63,71,25}. According to the above-mentioned benefits of the deployment of the RIS, it is considered as well-suited for integration into ICSC systems to enhance their overall performance and reliability. Consequently, researchers have been inspired to investigate the performance of RIS-assisted ICSC systems \cite{31,32}. In particular, \cite{31} utilized the RIS to alleviate interference between the radar and the edge computing server (ECS), proposing an energy-efficient system. Meanwhile, the authors of \cite{32} employed the RIS as a passive information source in an ICSC system and maximized the weighted throughput capacity, which was considered as a metric for performance evaluation.

However, both \cite{31} and \cite{32} considered ICSC systems deployed predominantly at the base station (BS) end, aligning with the BS-centric model where BSs are tasked with probing sensing signals, and UEs are responsible for offloading information to these BSs. This arrangement, as elaborated in \cite{22}, enhances efficient signal interference management and optimizes network resources. However, the BS-centric based ICSC systems may not fully cater to scenarios requiring rapid response and decentralized processing capabilities. In contrast, the User-centric based ICSC systems \cite{20,23,8} are capable of addressing these limitations by enabling UEs not only to engage in sensing tasks but also to process these tasks directly. This approach fosters high hardware sharing, thereby boosting spectrum and energy efficiency.

Owing to the well-suited nature of RIS to ICSC systems and the advantages of User-centric based ICSC systems, it is critical to investigate the RIS-assisted User-centric ICSC systems. This integration not only ensures low latency and high performance of ICSC systems but also contributes to boosting spectrum and energy efficiency. Specifically, RIS enhances the communication link of the ICSC system and facilitates the offloading of the substantial volume of data generated by the radar sensing to the ECS, thus ensuring the low latency attainment of the highly accurate sensing results. However, to the best of our knowledge, research on this integration is still missing.

While we recognize the significant importance, studying the RIS-assisted User-centric based ICSC systems is indeed a highly challenging task. Firstly, compared to the BS-centric\cite{22}, the User-centric approach inherently involves a wider variety of scenarios compared to the BS-centric model, increasing the difficulty due to the diversity and unpredictability of user environments and behaviors. Interference among UEs significantly affects the sensing and communication functionalities in ICSC systems, complicating the overall allocation of computational resources and the optimization of system performance. Additionally, when considering a more general case such as multi-input multi-output (MIMO), the challenges intensify. From a technical perspective, the incorporation of precoding in the form of matrices for MIMO scenarios presents significant challenges in processing the performance matrices in ICSC systems. This complexity is due to the need to optimize multiple spatial streams simultaneously, which requires intricate balancing of signal integrity across various channels and users within the User-centric ICSC framework. Furthermore, within this framework, it is crucial to address the challenge separately for both the multi-UE and single-UE scenarios, a topic that has not yet been explored in the literature. Thus, the methods proposed in prior works \cite{64,76} cannot be directly adopted.

Inspired by the preceding discussion, this paper delves into a framework for RIS-assisted ICSC systems with a User-centric focus, where the RIS is deployed to facilitate the offloading of the computing data to the ECS. For the multi-UE scenario, we propose a block coordinate descent (BCD)-based algorithm that integrates the minimum variance distortionless response (MVDR), successive convex approximation (SCA), majorization-minimization (MM), and weighted minimum mean-square error (WMMSE) to effectively address the challenges identified within the system. Then, for the simplified single-UE scenario, a low-complexity algorithm that yields closed-form solutions is presented.
The main contributions of this paper are summarized as follows:
\begin{itemize}
    \item \textbf{\textit{New User-centric Framework:}} An RIS-assisted ICSC system within a User-centric DFRC framework is investigated, where multiple DFRC-enabled UEs simultaneously detect targets and upload radar-sensing computational tasks to an ECS through their communication functionalities. Specifically, we jointly optimize the precoding and decoding matrices of the uplink signal, the receive radar beamforming vector, the RIS phase shift vector, the offloading volume, and the edge computational resource allocation to minimize the weighted latency. Due to the highly coupling among the variables, a BCD-based algorithm is proposed to effectively decouple the optimization of the computational settings and the active and passive beamforming matrices into two subproblems.
    \item \textbf{\textit{Comprehensive solution for multi-UE scenario:}} For the computational settings, we employ the BCD technique again to decouple the offloading volume and the edge computing resource allocation. When the allocation of edge computing resources is fixed, the optimal computational offloading can be determined under the assumption that the latencies of local and edge computing are equal. Given the computational offloading, the Karush-Kuhn-Tucker (KKT) conditions and the bisection search method are employed to determine the optimal allocation of edge computing resources. Then, with fixed computing design, we deal with the joint optimization problem of active and passive beamforming, the objective function (OF) of which takes the form of a non-convex sum-of-ratios, presenting a significant challenge. To address this issue, we reformulate the OF into a tractable form of a weighted sum-rate by introducing  auxiliary variables and subsequently develop an iterative algorithm to solve it. In each iteration, the auxiliary variables are updated using the modified Newton's method. Following the reformulation of the objective function with the WMMSE algorithm, the MVDR, SCA, and MM techniques are employed to optimize the the receive radar beamforming vector, the precoding matrices of the uplink signal and the RIS phase shift vector, respectively.
    \hyphenation{consi-dered}
    \item \textbf{\textit{Close-form Solution for Single-UE Scenario:}} Furthermore, a specific scenario involving a single UE is examined, where the edge computational resource allocation and interference from multiple UEs are not taken into account. We derive a closed-form solution in each iteration.
    \item \textbf{\textit{Effectiveness of the Framework:}} Simulation results verify the effectiveness of the proposed algorithms and the advantages of employing the RIS-assisted ICSC system. It is demonstrated that the proposed algorithm significantly improves the latency performance of the ISCC system in comparison to the conventional approach that lacks an integrated RIS.
\end{itemize}

\begin{figure}[t]
    \centering
    \includegraphics[width=2.6in]{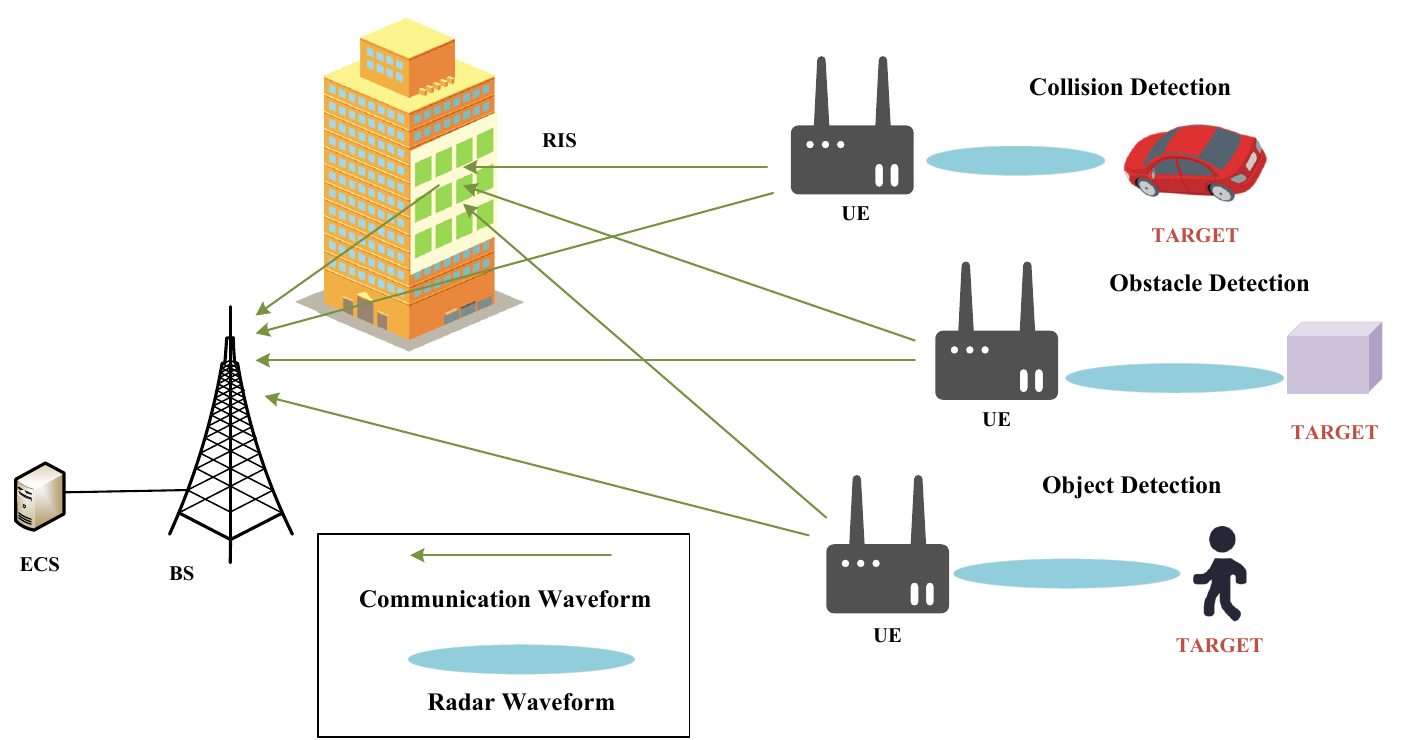}
    \caption{System model.}
    \label{fig1}
\end{figure}

\emph{Notation:} Throughout this paper, scalars, vectors, and matrices are denoted by lowercase italicized letters, lowercase bold letters, and uppercase bold letters, respectively. $\mathbb{C}^{N \times M}$ denotes the space of  $N \times M$ complex matrices. $(\cdot)^{\mathrm{T}}$, $(\cdot)^*$, and $(\cdot)^{\mathrm{H}}$ denote the transpose, complex conjugation, and Hermitian operators, respectively. $(\cdot)^{\star }$ represents the optimal result of a variable. $\mathrm{Tr}(\cdot)$ is the trace of a matrix. $\mathbb{E} (\cdot)$ denotes the statistical expectation operation. $\left\lfloor \cdot \right\rfloor $ and $\left\lceil \cdot \right\rceil $ are the floor and ceiling of a scalar, respectively. $\mathrm{Re} \{\cdot\} $ and $\angle (\cdot)$ represent the real part and the argument of a complex number. $\odot$ denotes the Hadamard product of two matrices.  $ \mathrm{diag} \{\bm{\mathrm{x}}\}$ means the diagonal matrix where the diagonal elements are $\bm{\mathrm{x}}$. $|\cdot|$ and $\left\lVert \cdot \right\rVert$ represent the absolute value of a scalar and the Euclidean norm of a vector, respectively.
\section{System Model and Problem Formulation}
\subsection{Communication Model}

\hyphenation{sa-tisfying}
Consider an RIS-assisted ICSC system as illustrated in Fig. \ref{fig1}, which consists of $K$ DFRC-enabled UEs with $N$ antennas, an $L$-element RIS, and an $M$-antenna BS connected to an ECS.\footnote{The radar sensing process generates a substantial volume of computation-intensive data, an ECS is situated at the BS to expedite the processing of the radar data.} Each UE is capable of simultaneously conducting target sensing and communicating with the BS to exchange crucial information and facilitate centralized control. The UE offloads data to the BS via the communication function and enable the rapid acquisition of sensing results with improved accuracy. Since UEs and their corresponding targets are spatially separated from the BS, it is assumed that the RIS is dedicated to providing support for edge computing and does not impact the sensing capabilities of the UEs\cite{18}.
By defining the precoding matrix as $\bm{\mathrm{F}}_{\mathrm{c},k} \in \mathbb{C}^{N\times d}$ and the transmit symbols at time slot $n$ as $\bm{\mathrm{c}}_k = [c_{k,1}[n], \cdots, c_{k,d}[n]]^{\mathrm{T}} \in \mathbb{C}^{d \times 1}$, satisfying  $\mathbb{E} [\bm{\mathrm{c}}_k\bm{\mathrm{c}}_k^{\mathrm{H}}]=\bm{\mathrm{I}}_d$ and $\mathbb{E} [\bm{\mathrm{c}}_k\bm{\mathrm{c}}_l^{\mathrm{H}}]=\bm{\mathrm{0}}_d$,  the transmit signal of UE $k$ is expressed as
\begin{equation}
    {\mathbf{x}}_k = {\mathbf{F}}_{\mathrm{c},k} \bm{\mathrm{c}}_k \in \mathbb{C}^{N \times 1}. \label{equ1}
\end{equation}

We denote the channels from UE $k$ to the BS, from the RIS to the BS, and from UE $k$ to the RIS by $\bm{\mathrm{H}}_{\mathrm{bu},k} \in \mathbb{C}^{M \times N}$, $\bm{\mathrm{H}}_{\mathrm{r}} \in \mathbb{C}^{M \times L}$, and $\bm{\mathrm{H}}_{\mathrm{ru},k} \in \mathbb{C}^{L \times N}$, respectively. For simplicity, the reflection elements of the RIS have unit amplitude and thus the reflecting coefficients matrix of the RIS can be expressed as $\bm{\mathrm{V}} = \mathrm{diag}\{e^{j\theta_{1}},\cdots,e^{j\theta _{L}}\} \in \mathbb{C}^{L \times L}$, where $\theta_l \in (0,2 \pi], \forall l \in \mathcal{L} = \{1,\cdots,L\}$ is the phase of the $l$-th reflecting element of RIS. The overall channel matrix between UE $k$ and the BS is given by
\begin{equation}
    \bm{\mathrm{H}}_k \triangleq \bm{\mathrm{H}}_{\mathrm{bu},k}+\bm{\mathrm{H}}_{\mathrm{r}} \bm{\mathrm{V}} \bm{\mathrm{H}}_{\mathrm{ru},k} 
    \in \mathbb{C}^{M \times N}.
\end{equation}
Then, the signal vector received at the BS is formulated as 
\begin{equation}
    \bm{\mathrm{y}}_{\mathrm{c}} = \sum ^K _{k = 1} \bm{\mathrm{H}}_k \bm{\mathrm{F}}_{\mathrm{c},k}\bm{\mathrm{c}}_k +
     \bm{\mathrm{n}}_{\mathrm{c}} \in \mathbb{C}^{M\times 1},
\end{equation}
where $\bm{\mathrm{n}}_{\mathrm{c}} = [n_{\mathrm{c},1},\cdots,n_{\mathrm{c},M}]^{\mathrm{T}} \in \mathbb{C}^{M \times 1} $, which satisfies $\bm{\mathrm{n}}_{\mathrm{c}}\sim \mathcal{C} \mathcal{N} (\bm{\mathrm{0}},\sigma_{\mathrm{c}}^2\bm{\mathrm{I}}_M)$, denotes the additive white Gaussian noise (AWGN) at the BS. 
After decoding the received signal with a linear matrix $\bm{\mathrm{W}}_{\mathrm{c},k}$, we can obtain the recovered signal of UE $k$ as
\begin{equation}
    \begin{aligned}
        \widehat{\bm{\mathrm{c}}}  _{k} &=\bm{\mathrm{W}}_{\mathrm{c},k}^{\mathrm{H}}  \bm{\mathrm{y}}_{\mathrm{c}} \\
    &= \bm{\mathrm{W}}_{\mathrm{c},k}^{\mathrm{H}} (\sum ^K _{k = 1} \bm{\mathrm{H}}_k \bm{\mathrm{F}}_{\mathrm{c},k}\bm{\mathrm{c}}_k + {\mathbf{n}}_{\mathrm{c}} ) \in \mathbb{C}^{d \times 1}.
    \end{aligned}
\end{equation}
The covariance of the transmit signal from UE $k$ is given by
\begin{equation}
    \bm{\mathrm{R}}_{\bm{\mathrm{X}},k}\triangleq \mathbb{E} [\bm{\mathrm{x}}_k\bm{\mathrm{x}}_k^{\mathrm{H}}]=\bm{\mathrm{F}}_{\mathrm{c},k}\bm{\mathrm{F}}_{\mathrm{c},k}^{\mathrm{H}}.
\end{equation} 
By denoting the transmit power of UE $k$ by $P_{\mathrm{t},k}$, the power constraint for UE $k$ is expressed as
\begin{equation}
    \mathrm{Tr}(\bm{\mathrm{R}}_{\bm{\mathrm{X}},k})= \mathrm{Tr}(\bm{\mathrm{F}}_{\mathrm{c},k}\bm{\mathrm{F}}_{\mathrm{c},k}^{\mathrm{H}})  \leqslant P_{\mathrm{t},k}, \forall k \in \mathcal{K} = \{1,\cdots,K\}.
\end{equation}
Furthermore, the achievable off-loading rate (bit/s/Hz) of UE $k$ can be formulated as
    \begin{align}
&R_k(\bm{\mathrm{F}}_{\mathrm{c}},\bm{\mathrm{W}}_{\mathrm{c},k},\bm{\mathrm{\theta }}) \label{equ7} \\
        &= B_\omega \log_2 |\bm{\mathrm{I}}_d+\bm{\mathrm{W}}_{\mathrm{c},k}^{\mathrm{H}} \bm{\mathrm{H}}_k\bm{\mathrm{F}}_{\mathrm{c},k}\bm{\mathrm{F}}_{\mathrm{c},k}^{\mathrm{H}}\bm{\mathrm{H}}_k^{\mathrm{H}} \bm{\mathrm{W}}_{\mathrm{c},k} (\bm{\mathrm{W}}_{\mathrm{c},k}^{\mathrm{H}} \bm{\mathrm{J}}_k \bm{\mathrm{W}}_{\mathrm{c},k})^{-1}|, \nonumber
    \end{align}
where $B_\omega$ denotes the bandwidth, $\bm{\mathrm{\theta }}$ is defined as $\bm{\mathrm{\theta }} =[\theta _{1},\cdots,\theta _{L}]$, and $\bm{\mathrm{J}}_k$ represents the interference plus noise matrix of UE $k$ that is given by
\begin{equation}
\begin{aligned}
    \nonumber
    \bm{\mathrm{J}}_k =  \sum ^K _{i = 1,i\neq k}  \bm{\mathrm{H}}_i\bm{\mathrm{F}}_{\mathrm{c},i}\bm{\mathrm{F}}_{\mathrm{c},i}^{\mathrm{H}}\bm{\mathrm{H}}_i^{\mathrm{H}} +\sigma_{\mathrm{c}} ^2 \bm{\mathrm{I}}_M \in \mathbb{C}^{M \times M}.
\end{aligned}
\end{equation}
\vspace{-0.4cm}
\subsection{Radar Sensing Model}
Given the transmit signal $\bm{\mathrm{x}}_k $ in   (\ref{equ1}), the received echo of UE $k$ is expressed as\footnote{Note that the radar clutter can be effectively mitigated or eliminated through radar signal processing techniques, we disregard the effects of clutter as \cite{8}. Thus, the radar and computational offloading performance achieved by the proposed algorithms can be regarded as theoretical upper bounds.}
\begin{equation}
    \begin{aligned}
        \bm{\mathrm{y}}_{\mathrm{s},k} =\alpha _k \bm{\mathrm{a}}_{\mathrm{R},k}(\theta _k)\bm{\mathrm{a}}_{\mathrm{T},k}^{\mathrm{H}}(\theta _k) \bm{\mathrm{x}}_k +\sum ^K _{i = 1,i\neq k}  \bm{\mathrm{H}}_{k,i}\bm{\mathrm{x}}_i +  \bm{\mathrm{z}}_k \in \mathbb{C}^{N \times 1},
    \end{aligned}
\end{equation}
where $\alpha _k$ represents the amplitude of the radar target located at angle $\theta_k$, $\bm{\mathrm{a}}_{\mathrm{R},k}(\theta _k) \in \mathbb{C}^{N \times 1} $ and $\bm{\mathrm{a}}_{\mathrm{T},k}(\theta _k) \in \mathbb{C}^{N \times 1} $ are the receive and transmit array response vectors for UE $k$, respectively. Matrix $\bm{\mathrm{H}}_{k,i} \in \mathbb{C}^{N \times N}$ denotes the channel matrix between UE $k$ and UE $i$, and $\bm{\mathrm{z}}_k = [z_{k,1},\cdots,z_{k,N}]^{\mathrm{T}} \in \mathbb{C}^{N \times 1}$ is the AWGN with the covariance of $\sigma_{\mathrm{s}}^2$. For simplicity, we define $\bm{\mathrm{G}}_{k} \triangleq  \alpha _k \bm{\mathrm{a}}_{\mathrm{R},k}(\theta _k)\bm{\mathrm{a}}_{\mathrm{T},k}^{\mathrm{H}}(\theta _k) \in \mathbb{C}^{N \times N}$ as the target response matrix. Then, the radar output at UE $k$ is given by
\begin{equation}
    \begin{aligned}
        r_k 
        &= \bm{\mathrm{w}}_{\mathrm{s},k}^{\mathrm{H}}\bm{\mathrm{G}}_k \bm{\mathrm{x}}_k +\sum ^K _{i = 1,i\neq k} \bm{\mathrm{w}}_{\mathrm{s},k}^{\mathrm{H}} \bm{\mathrm{H}}_{k,i}\bm{\mathrm{x}}_i + \bm{\mathrm{w}}_{\mathrm{s},k}^{\mathrm{H}} \bm{\mathrm{z}}_k,
    \end{aligned}
\end{equation}
where $\bm{\mathrm{w}}_{\mathrm{s},k} \in \mathbb{C}^{N \times 1}$ is the receive beamforming vector of UE $k$.
 
The sensing performance can be measured in terms of the radar SINR. After applying receive beamforming, the radar SINR for UE $k$, denoted by  $\gamma_k$, is given by
\begin{equation}
    \label{equs}
    \begin{aligned}
        \gamma _k & = \frac{ \mathbb{E}\left[| \bm{\mathrm{w}}_{s,k}^{\mathrm{H}} \bm{\mathrm{G}}_k \bm{\mathrm{x}}_k | ^2\right]}
    {\mathbb{E} [ |\bm{\mathrm{w}}_{s,k}^{\mathrm{H}}(\sum ^K _{i = 1,i\neq k} \bm{\mathrm{H}}_{k,i}\bm{\mathrm{x}}_i + \bm{\mathrm{z}}_k) | ^2]}  \\
    &= \frac{\mathrm{Tr}(\bm{\mathrm{w}}_{s,k}^{\mathrm{H}} \bm{\mathrm{G}}_k \bm{\mathrm{F}}_{\mathrm{c},k}\bm{\mathrm{F}}_{\mathrm{c},k} ^{\mathrm{H}} \bm{\mathrm{G}}_k^{\mathrm{H}} \bm{\mathrm{w}}_{s,k})}
    {\mathrm{Tr}(\bm{\mathrm{w}}_{s,k}^{\mathrm{H}}\bm{\mathrm{T}}_k\bm{\mathrm{w}}_{s,k})},
    \end{aligned}
\end{equation}
where $\bm{\mathrm{T}}_k \triangleq  \sum ^K _{i = 1,i\neq k}\bm{\mathrm{H}}_{k,i}  \bm{\mathrm{F}}_{\mathrm{c},i} \bm{\mathrm{F}}_{\mathrm{c},i} ^{\mathrm{H}} \bm{\mathrm{H}}_{k,i}^{\mathrm{H}} +\sigma_{\mathrm{s}}^2\bm{\mathrm{I}}_N \in \mathbb{C}^{N \times N}$.

\subsection{Computation Model}
\label{secc}
The partial off-loading strategy is adopted for computational tasks, where a fraction of the data is processed locally, while the remaining portion is offloaded to the BS. 

\paragraph{Local Computation Model}
For UE $k$, we denote the overall number of bits awaiting processing, the number of bits designated for offloading, the number of CPU cycles required for processing each bit, and its CPU frequency by $V_k$, $v_k$, $c_k$, and $f_k^l$, respectively. Then, the latency for local computation of UE $k$ is given by $T_{\mathrm{l},k} = (V_k-v_k)c_k/f_k^l$.

\hyphenation{compu-ting}
\paragraph{Edge Computation Model}
Let $f_k^e$ with $\sum _{k=1}^Kf_k^e\leqslant f_{\mathrm{total}}^e$ denote the computation resources allocated to UE $k$, where $f_{\mathrm{total}}^e$ represents the processing rate of the ECS. It is assumed that the edge computing process for UE $k$ begins when all $v_k$ offloaded bits are received at the BS. The entire process of edge computation for UE $k$ can be divided into three steps: computing offloading, edge computing, and result feedback, with the feedback latency typically being neglected \cite{70}. Consequently, the total latency of edge computation for UE $k$ is expressed as $T_{\mathrm{c},k} = T_{\mathrm{o},k} +T_{\mathrm{e},k} = v_k/R_k(\bm{\mathrm{F}}_{\mathrm{c}},\bm{\mathrm{W}}_{\mathrm{c},k},\bm{\mathrm{\theta }}) + v _k c_k/ f_k^e$.

We consider the simultaneous execution for local and edge computation. In this case, the latency of UE $k$ can be represented by
\begin{equation}
    \begin{aligned}
        &T_k(\bm{\mathrm{F}}_{\mathrm{c}},\bm{\mathrm{W}}_{\mathrm{c},k},\bm{\mathrm{\theta }},v _k,f_k^e) = \max \{T_{\mathrm{l},k} ,T_{\mathrm{c},k}\}\\
        &= \max\{\frac{(V _k-v _k)c_k}{f_k^l},\frac{v _k}{R_k(\bm{\mathrm{F}}_{\mathrm{c}},\bm{\mathrm{W}}_{\mathrm{c},k},\bm{\mathrm{\theta }})}+\frac{v _k c_k}{f_k^e} \}.
    \end{aligned}
\end{equation}
\subsection{Problem Formulation}
\hyphenation{a-llocation}
In this paper, we aim to minimize the total weighted  latency of all UEs by jointly optimizing the communication precoding matrices $\bm{\mathrm{F}}_{\mathrm{c}}=\{\bm{\mathrm{F}}_{\mathrm{c},k},\forall k \in \mathcal{K}\}$, the communication decoding matrices $\bm{\mathrm{W}}_{\mathrm{c}} = \{\bm{\mathrm{W}}_{\mathrm{c},k},\forall k \in \mathcal{K}\}$, the receive radar beamforming vectors $\bm{\mathrm{w}}_{\mathrm{s}} = \{\bm{\mathrm{w}}_{\mathrm{s},k},\forall k \in \mathcal{K}\}$, the RIS phase shift vector $\bm{\mathrm{\theta }} =[\theta _{1},\cdots,\theta _{L}]$, the offloading volume $\bm{\mathrm{v}} = [v _1,\cdots,v _K]^{\mathrm{T}}$, and the edge computational resource allocation $\bm{\mathrm{f}} = [f_1^e,\cdots,f_K^e]^{\mathrm{T}}$. The optimization problem is formulated as
\begin{subequations}
    \begin{align}
        \mathcal{P} 0: \min_{\bm{\mathrm{F}}_{\mathrm{c}},\bm{\mathrm{W}}_{\mathrm{c}},\bm{\mathrm{w}}_{\mathrm{s}},\bm{\mathrm{\theta }},\bm{\mathrm{v}},\bm{\mathrm{f}}}  
        &\sum _{k=1}^K \xi  _k T_k(\bm{\mathrm{F}}_{\mathrm{c}},\bm{\mathrm{W}}_{\mathrm{c},k},\bm{\mathrm{\theta }},v _k,f_k^e)  \\
        s.t. \quad \quad
        & \mathrm{Tr}(\bm{\mathrm{R}}_{\bm{\mathrm{X}},k}) \leqslant P_{t,k},\forall k \in \mathcal{K}, \label{eqa} \\
        & \gamma _k \geqslant  \eta ,\forall k \in \mathcal{K} ,\label{eqb}\\
        & 0 < \theta _l \leqslant 2\pi ,\forall l \in \mathcal{L},\label{eqd}\\
        &\sum _{k=1}^Kf_k^e \leqslant f_{\mathrm{total}}^e, \label{eqc} \\
        & v _k \in \{0,1,\cdots,V _k\},\forall k \in \mathcal{K} , \label{eqe} \\
        & f_k^e \geqslant 0, \forall k \in \mathcal{K}, \label{eqf}
    \end{align}
\end{subequations}
where $\xi_k$ represents the weight of UE $k$ and $\eta$ denotes the radar SINR requirement of UEs. 
The constraints (\ref{eqc}) - (\ref{eqf}) are introduced based on the computation model in Sec. \ref{secc}.

Addressing Problem $\mathcal{P} 0$ poses several challenges. Firstly, the OF of Problem $\mathcal{P} 0$ is in a segmented form. Secondly, the variables of Problem $\mathcal{P} 0$ are highly coupled. Finally, the weighted total latency involves calculating the sum of ratios, which can often be challenging to optimize. 
\section{Algorithm Design}
\label{secAD}
In this section,  Problem $\mathcal{P} 0$ is decoupled into multiple subproblems by using the BCD method, for which solutions are subsequently devised.
\subsection{Joint Optimization of Offloading Volume $\bm{\mathrm{v}}$ and Edge Computing Resource Allocation  $\bm{\mathrm{f}}$}
\label{secv}
Given $\{\bm{\mathrm{F}}_{\mathrm{c}}, \bm{\mathrm{W}}_{\mathrm{c}}, \bm{\mathrm{w}}_{\mathrm{s}}, \bm{\theta}\}$, Problem $\mathcal{P} 0 $ is reformulated as
\begin{equation}
    \begin{aligned}
        \label{equ13}
        \mathcal{P} 1:\min_{\bm{\mathrm{v}},\bm{\mathrm{f}}} \quad& \sum _{k=1}^K \xi  _k T_k(v _k,f_k^e)  \\
        s.t. \quad&
        (\ref{eqc}),(\ref{eqe}),(\ref{eqf}) .
    \end{aligned}
\end{equation}
It can be verified that $T_{\mathrm{l},k}$ and $T_{\mathrm{c},k}$ are respectively monotonically decreasing and increasing with respect to (w.r.t.) $v _k$. Therefore, given $\bm{\mathrm{f}}$, the minimum value of $T_k(v _k) = \max \{T_{\mathrm{l},k} ,T_{\mathrm{c},k}\}$ is obtained when $T_{\mathrm{l},k} = T_{\mathrm{c},k}$ is satisfied. Accordingly, the value of $v_k$ that minimizes $T_k(v_k)$ without the integer constraint, denoted by $\hat{v _k}^{\star}$, is given by
\begin{equation}
    \label{equ15}
    \hat{v _k}^{\star}  = \frac{V_k c_k R_k f_k^e}{f_k^e f_k^l + c_k R_k(f_k^e + f_k^l)}, \forall k \in \mathcal{K} .
\end{equation}
Therefore, the optimal value of $v _k$ is obtained according to 
\begin{equation}
    v _k^{\star } = \arg \min _{\hat{v _k}\in \{ \left\lfloor \hat{v _k}^{\star} \right\rfloor, \left\lceil \hat{v _k}^{\star} \right\rceil  \}} 
    T_k(\hat{v _k}), \forall k \in \mathcal{K}.
\end{equation} 

Given $\bm{\mathrm{v}}$, Problem $\mathcal{P} 1$ can be rewritten as follows:
\begin{equation}
    \begin{aligned}
        \mathcal{P} 1\text{-}1: \min_{\bm{\mathrm{f}}} \quad&\sum _{k=1}^K \xi  _k  \frac{V_k c_k^2 R_k + V_k c_k f_k^e}{f_k^e f_k^l + c_k R_k(f_k^e + f_k^l)} \\
        s.t. \quad&(\ref{eqc}),(\ref{eqf}).
        \label{equ16}
    \end{aligned}
\end{equation}
Problem $\mathcal{P} 1\text{-}1$ is convex w.r.t. $\bm{\mathrm{f}}$ and satisfies the Slater's condition\cite{44}. Hence, the optimal solution of Problem $\mathcal{P} 1\text{-}1$ can be derived via the KKT conditions. Specifically, the Lagrangian function of Problem $\mathcal{P} 1\text{-}1$ can be formulated as  
\begin{equation}
    \mathcal{L} (\bm{\mathrm{f}},\mu ) = \sum _{k=1}^K \xi  _k  \frac{V_k c_k^2 R_k + V_k c_k f_k^e}{f_k^e f_k^l + c_k R_k(f_k^e + f_k^l)} 
    + \mu(\sum _{k=1}^K f_k^e - f_{\mathrm{total}}^e),
\end{equation}
where $\mu \geqslant  0$ denotes the Lagrange multiplier. Furthermore, the KKT conditions for Problem $\mathcal{P} 1\text{-}1$ are expressed as 
\begin{align}
& \frac{\partial \mathcal{L}}{\partial f_k^e} = \frac{-\xi V_k c_k^3 R_k ^2} {[ c_k R_k f_k^l +(f_k^l + c_k R_k){f_k^{e }}^{ \star}]^2 } + \mu ^{\star } = 0, \label{equ18} \\ 
& \mu^{\star } (\sum _{k=1}^K {f_k^{e }}^{ \star} - f_{\mathrm{total}}^e) = 0, \label{equ19} \\  
& {f_k^{e }}^{ \star} \geqslant 0. \label{equ20} 
\end{align}
Given $\mu$, we can obtain the value of ${f_k^{e }}^{ \star}$ by solving (\ref{equ18}) as
\begin{equation}
    \label{equ24}
    {f_k^{e }}^{ \star} = \frac{\sqrt{\frac{\xi V_k c_k^3 R_k ^2}{\mu^{\star }}}-c_k R_k f_k^l }{f_k^l + c_k R_k}, \forall k \in \mathcal{K} .
\end{equation}
From (\ref{equ20}), we have $\sqrt{\frac{\xi V_k c_k^3 R_k ^2}{\mu^{\star }}}-c_k R_k f_k^l \geqslant 0 $, and thus $ \mu^{\star } \leqslant  \frac{\xi V_k c_k}{f_k^{l2}}$. Consequently, the value $\mu^{\star }$ is within the range of $ (\mu_{l}, \mu_{u}] = (0,{\min}_k \frac{\xi V_k c_k}{{f_k^{l}}^2}] $. As $ \sum _{k=1}^Kf_k^e $ monotonically decreases w.r.t. $ \mu $, the optimal solution for $f_k^e$ can be obtained by performing the bisection search based on (\ref{equ19}). The algorithm for solving Problem $\mathcal{P} 1$ is outlined in Algorithm \ref{alg:algorithm1}. The complexity of Algorithm \ref{alg:algorithm1} is mainly dependent on the bisection search method. Given the convergence criterion $\epsilon $ and the maximum iteration time for variables $\{\bm{\mathrm{v}},\bm{\mathrm{f}} \}$, denoted by  $t_1^{\max}$, the overall complexity of Algorithm \ref{alg:algorithm1} can be expressed as $\mathcal{O} (t_1^{\max} \log _2 (\frac{\mu_{u}-\mu_{l}}{\epsilon})K)$. 
\begin{algorithm}[t]
    \caption{Joint optimization of $\{\bm{\mathrm{v}},\bm{\mathrm{f}} \}$, given $\{\bm{\mathrm{F}}_{\mathrm{c}},\bm{\mathrm{W}}_{\mathrm{c}},\bm{\mathrm{w}}_{\mathrm{s}},\bm{\mathrm{\theta }} \}$}
    \label{alg:algorithm1}
    \begin{algorithmic}[1]
      \State Initialize iteration number $t_1 = 0$. Set the maximum number of iterations to $t_1^{\max}$. Initialize $\bm{\mathrm{f}}^{(0)}$ satisfying (\ref{eqc}) and (\ref{eqf}).
      \State  \textbf{Repeat} 
      \State \quad Calculate $\bm{\mathrm{v}}^{(t_1+1)}$ using (\ref{equ15}).
      \State \quad Calculate $\mu $ and $\bm{\mathrm{f}}^{(t_1+1)}$ accordingly  using the bisection search method. \label{step4}
        \State \quad $ t_1 \gets t_1+1 $.
        \State \textbf{Until} $t_1=t_{1}^{\rm{max}}$ \textbf{or}
        \begin{small}
            \begin{align}
                \epsilon _1^{(t_1+1)} = \frac{\left\lvert \mathrm{obj}^{(t_1+1)}(\bm{\mathrm{v}},\bm{\mathrm{f}}) - \mathrm{obj}^{(t_1)}(\bm{\mathrm{v}},\bm{\mathrm{f}}) \right\rvert }{\mathrm{obj}^{(t_1+1)}(\bm{\mathrm{v}},\bm{\mathrm{f}})}< \epsilon. \nonumber
            \end{align}
        \end{small}
    \end{algorithmic}
  \end{algorithm}    
\subsection{Joint Optimization of Passive and Active Beamforming $\{\bm{\mathrm{F}}_{\mathrm{c}},\bm{\mathrm{W}}_{\mathrm{c}},\bm{\mathrm{w}}_{\mathrm{s}},\bm{\theta }\}$}
\subsubsection{Problem Reformulation}
With the optimized offloading volume $\bm{\mathrm{v}}$, minimizing $T_k$ is equivalent to minimizing the computing offloading latency. Therefore, given $\{\bm{\mathrm{v}},\bm{\mathrm{f}}\}$, Problem $\mathcal{P} 0$ is transformed into
\begin{equation}
    \label{equ26}
    \begin{aligned}
        \mathcal{P} 2:\min_{\bm{\mathrm{F}}_{\mathrm{c}},\bm{\mathrm{W}}_{\mathrm{c}},\bm{\mathrm{w}}_{\mathrm{s}},\bm{\mathrm{\theta}}}  & 
        \sum _{k=1}^K \xi  _k \frac{v _k}{R_k(\bm{\mathrm{F}}_{\mathrm{c}},\bm{\mathrm{W}}_{\mathrm{c},k},\bm{\mathrm{\theta}})}  \\
        s.t.\quad \, &(\ref{eqa}),(\ref{eqb}),(\ref{eqd}). 
    \end{aligned}
\end{equation}
The OF of Problem $\mathcal{P} 2$ is in a non-convex sum-of-ratios form, thereby presenting challenges in deriving solutions through the existing methods\cite{56}. To address this issue, we first reformulate Problem $\mathcal{P} 2$ by introducing an auxiliary variable $\bm{\lambda}=[\lambda_1,\cdots,\lambda_K]$ as 
\begin{equation}
    \label{equ27}
    \begin{aligned}
        \mathcal{P} 2\text{-}1:\min_{\bm{\mathrm{F}}_{\mathrm{c}},\bm{\mathrm{W}}_{\mathrm{c}},\bm{\mathrm{w}}_{\mathrm{s}},\bm{\mathrm{\theta}},\bm{\mathrm{\lambda }}}& \,  \sum _{k=1}^K \lambda _k \\
        s.t. \quad \ \
        &\frac{ \xi _k v _k}{R_k(\bm{\mathrm{F}}_{\mathrm{c}},\bm{\mathrm{W}}_{\mathrm{c},k},\bm{\mathrm{\theta}})} \leqslant \lambda _k,\forall k \in \mathcal{K},  \\
        &(\ref{eqa}),(\ref{eqb}),(\ref{eqd}). 
    \end{aligned}
\end{equation}
Then, we introduce the following proposition for further derivation. 
\begin{proposition}
\label{pro1}
    If $\{\bm{\mathrm{F}}_{\mathrm{c}}^{\star},\bm{\mathrm{W}}_{\mathrm{c}}^{\star},\bm{\mathrm{w}}_{\mathrm{s}}^{\star },\bm{\theta }^{\star}, \bm{\lambda}^{\star}\}$ is the solution of Problem $\mathcal{P} 2\text{-}1$, then $\bm{\delta}^ {\star} = [\delta _1,\cdots,\delta _K]$ exists such that $\{\bm{\mathrm{F}}_{\mathrm{c}}^{\star},\bm{\mathrm{W}}_{\mathrm{c}}^{\star},\bm{\mathrm{w}}_{\mathrm{s}}^{\star },\bm{\theta }^{\star}\}$ satisfies the KKT conditions of the following problem when $\bm{\lambda} = \bm{\lambda}^{\star}$ and $\bm{\delta} = \bm{\delta}^{\star}$
    \begin{equation}
        \label{equ28}
        \begin{aligned}
            \mathcal{P} 2\text{-}2:\min_{\bm{\mathrm{F}}_{\mathrm{c}},\bm{\mathrm{W}}_{\mathrm{c}},\bm{\mathrm{w}}_{\mathrm{s}},\bm{\mathrm{\theta}}}& \, \sum _{k=1}^K \delta _k[\xi _k v _k - \lambda _k R_k(\bm{\mathrm{F}}_{\mathrm{c}},\bm{\mathrm{W}}_{\mathrm{c},k},\bm{\mathrm{\theta}})]\\
            s.t.\quad \
            &(\ref{eqa}),(\ref{eqb}),(\ref{eqd}).  
        \end{aligned}
    \end{equation}
    Furthermore, $\{\bm{\mathrm{F}}_{\mathrm{c}}^{\star},\bm{\mathrm{W}}_{\mathrm{c}}^{\star},\bm{\theta }^{\star}\}$ satisfies the following equations when $\bm{\lambda} = \bm{\lambda}^{\star}$ and $\bm{\delta} = \bm{\delta}^{\star}$
    \begin{equation}
        \label{equ29}
        \begin{cases}
            \delta _k^{\star} = \frac{1}{R_k(\bm{\mathrm{F}}_{\mathrm{c}}^{\star},\bm{\mathrm{W}}_{\mathrm{c},k}^{\star},\bm{\mathrm{\theta}}^{\star})},\forall k \in \mathcal{K}, \\
            \lambda _k^{\star} = \frac{\xi _k v _k }{R_k(\bm{\mathrm{F}}_{\mathrm{c}}^{\star},\bm{\mathrm{W}}_{\mathrm{c},k}^{\star},\bm{\mathrm{\theta}}^{\star})},\forall k \in \mathcal{K}. \\
        \end{cases}
    \end{equation}
\end{proposition}
\begin{proof}
    Please refer to Appendix \ref{app1}.
\end{proof}
From  Proposition \ref{pro1}, if $(\bm{\mathrm{F}}_{\mathrm{c}}^{\star},\bm{\mathrm{W}}_{\mathrm{c}}^{\star},\bm{\mathrm{w}}_{\mathrm{s}}^{\star}, \bm{\mathrm{\theta}}^{\star})$ is the solution to Problem $\mathcal{P} 2\text{-}2$, and the value of $\bm{\mathrm{\lambda }}$ and $\bm{\mathrm{\delta}}$ are set following (\ref{equ29}), then $(\bm{\mathrm{F}}_{\mathrm{c}}^{\star},\bm{\mathrm{W}}_{\mathrm{c}}^{\star},\bm{\mathrm{w}}_{\mathrm{s}}^{\star},\bm{\mathrm{\theta}}^{\star},\bm{\mathrm{\lambda }}^{\star})$ is the solution to Problem $\mathcal{P} 2\text{-}1$. In the following, we propose an algorithm to solve Problem $\mathcal{P} 2\text{-}3$ by alternately optimizing the sets $\{\bm{\lambda},\bm{\delta }\}$ and $\{\bm{\mathrm{F}}_{\mathrm{c}}, \bm{\mathrm{W}}_{\mathrm{c}},\bm{\mathrm{w}}_{\mathrm{s}},\bm{\theta }\}$. Firstly, given $\{\bm{\mathrm{F}}_{\mathrm{c}}, \bm{\mathrm{W}}_{\mathrm{c}},\bm{\mathrm{w}}_{\mathrm{s}},\bm{\theta }\}$, $\bm{\lambda}$ and $\bm{\delta }$ are updated by using the modified Newton's method\cite{25,45}. Then, we optimize $\bm{\mathrm{F}}_{\mathrm{c}}$, $\bm{\mathrm{W}}_{\mathrm{c}}$, $\bm{\mathrm{w}}_{\mathrm{s}}$, and $\bm{\theta }$ given $\{\bm{\lambda},\bm{\delta }\}$ according to Algorithm \ref{alg:algorithm3}, which is discussed in the subsequent subsections. 
The procedure for solving Problem $\mathcal{P} 2\text{-}2$ is summarized in Algorithm \ref{alg:algorithm2}, where we have
    \begin{align}
        &\chi_k (\delta _k) = \delta _k R_k(\bm{\mathrm{F}}_{\mathrm{c}}^{\star},\bm{\mathrm{W}}_{\mathrm{c},k}^{\star},\bm{\mathrm{\theta}}^{\star})-1,\forall k \in \mathcal{K},\\
        &\kappa _k (\lambda _k) = \lambda _kR_k(\bm{\mathrm{F}}_{\mathrm{c}}^{\star},\bm{\mathrm{W}}_{\mathrm{c},k}^{\star},\bm{\mathrm{\theta}}^{\star})-\xi _k v _k, \forall k \in \mathcal{K}.
    \end{align}
    \begin{algorithm}
        \caption{Joint optimization of $\{\bm{\mathrm{F}}_{\mathrm{c}},\bm{\mathrm{W}}_{\mathrm{c}},\bm{\mathrm{w}}_{\mathrm{s}},\bm{\mathrm{\theta }}\}$, given $\{\bm{\mathrm{v}},\bm{\mathrm{f}}\}$.}
        \label{alg:algorithm2}
        \begin{algorithmic}[1]
          \State Initialize iteration number $t_2 = 0$ and step $\zeta \in (0,1)$. Set the maximum number of iterations to $t_2^{\max}$. Calculate $R_k^{(0)}$ using (\ref{equr}), $\bm{\mathrm{\delta}} ^{(0)}$ and $ \bm{\mathrm{\lambda }} ^{(0)}$ using (\ref{equ29}).
          \State  \textbf{Repeat}  
          \State \quad Calculate $\{\bm{\mathrm{F}}_{\mathrm{c}}^{(t_2+1)},\bm{\mathrm{W}}_{\mathrm{c}}^{(t_2+1)},\bm{\mathrm{w}}_{\mathrm{s}}^{(t_2+1)}, \bm{\mathrm{\theta }}^{(t_2+1)}\}$ using Algorithm \ref{alg:algorithm3}.
          \State \quad Update $\delta _k^{(t_2 +1)} $ and $\lambda _k ^{(t_2 +1)}$ by 
          \begin{small}
            \begin{align}
                & \delta _k^{(t_2 +1)} = \delta _k^{(t_2)} - \frac{\zeta ^{i^{(t_2 +1)}} \chi _k(\delta _k^{(t_2)})}{R_k^{(t_2+1)}(\bm{\mathrm{F}}_{\mathrm{c}},\bm{\mathrm{W}}_{\mathrm{c}},\bm{\mathrm{\theta}})}, \\
                & \lambda _k^
                {(t_2 +1)} = \lambda _k^{(t_2)} - \frac{\zeta ^{i^{(t_2 +1)}} \kappa  _k( \lambda  _k^{(t_2)})}{R_k^{(t_2+1)}(\bm{\mathrm{F}}_{\mathrm{c}},\bm{\mathrm{W}}_{\mathrm{c}},\bm{\mathrm{\theta}})},
            \end{align}
          \end{small}
            where $i^{(t_2+1)}$ is the smallest integer satisfying
            \begin{small}
                \begin{align}
                    &\sum_{k=1}^K \left\lvert \chi _k \left( \delta _k^{(t_2)} - \frac{\zeta ^{i} \chi _k(\delta _k^{(t_2)})}{R_k^{(t_2+1)}(\bm{\mathrm{F}}_{\mathrm{c}},\bm{\mathrm{W}}_{\mathrm{c}},\bm{\mathrm{\theta}})} \right) \right\rvert^2 + \nonumber\\
                  &\sum_{k=1}^K \left\lvert \chi _k \left( \lambda _k^{(t_2)} - \frac{\zeta ^{i} \kappa  _k( \lambda  _k^{(t_2)})}{R_k^{(t_2+1)}(\bm{\mathrm{F}}_{\mathrm{c}},\bm{\mathrm{W}}_{\mathrm{c}},\bm{\mathrm{\theta}})} \right) \right\rvert^2 \\
                &\leqslant (1-\epsilon _3 \zeta ^i)^2 \sum_{k=1}^K \left[ \left\lvert \chi_k (\delta _k^{(t_2+1)})  \right\rvert ^2 + \left\lvert \kappa _k (\lambda _k^{(t_2 + 1)})  \right\rvert ^2   \right]. \nonumber 
                \end{align}
            \end{small}
            \State \quad $t_2 \gets t_2 + 1$.
            \State \textbf{Until} $\chi_k (\delta _k^{(t_2)})= 0$ \textit{and} $\kappa _k (\lambda _k^{(t_2)}) = 0$ \textbf{or} $t_2=t_{2}^{\rm{max}}$.
        \end{algorithmic}
    \end{algorithm}
\subsubsection{Joint Active and Passive Beamforming Optimization}
In this case, Problem $\mathcal{P} 2\text{-}2$ can be transformed into 
\begin{equation}
    \label{equq}
    \begin{aligned}
        \mathcal{P} 2\text{-}3:\max_{\bm{\mathrm{F}}_{\mathrm{c}},\bm{\mathrm{W}}_{\mathrm{c}},\bm{\mathrm{w}}_{\mathrm{s}}, \bm{\mathrm{\theta}}} &\sum _{k=1}^K \delta _k\lambda _k R_k(\bm{\mathrm{F}}_{\mathrm{c}},\bm{\mathrm{W}}_{\mathrm{c},k},\bm{\mathrm{\theta}})\\
         s.t.\quad \, & (\ref{eqa}),(\ref{eqb}),(\ref{eqd}).
    \end{aligned}
\end{equation}
\paragraph{OF Reformulation and Decoding Matrix Design}
By introducing an auxiliary matrix $ \bm{\mathrm{D}}_k$, the equivalent formulation of Problem $\mathcal{P} 2\text{-}3$ with constants removed can be derived as \cite{33}
\begin{equation}
    \label{equ30}
    \begin{aligned}       
        \mathcal{P} 2\text{-}4:\min_{\bm{\mathrm{F}}_{\mathrm{c}},\bm{\mathrm{W}}_{\mathrm{c}},\bm{\mathrm{w}}_{\mathrm{s}}, \bm{\mathrm{\theta}},\bm{\mathrm{D}}} &\sum _{k=1}^K \delta _k\lambda _k\mathrm{Tr}(\bm{\mathrm{D}}_k \bm{\mathrm{E}}_k(\bm{\mathrm{F}}_{\mathrm{c}},\bm{\mathrm{W}}_{\mathrm{c},k},\bm{\mathrm{\theta}})) \\
        &-\delta _k\lambda _k \log _2|\bm{\mathrm{D}}_k|\\
        s.t.\quad \ \
        & (\ref{eqa}),(\ref{eqb}),(\ref{eqd}),
    \end{aligned}
\end{equation}
where $ \bm{\mathrm{D}} = \{\bm{\mathrm{D}}_k, \forall k \in \mathcal{K} \}$ and $\bm{\mathrm{E}}_k(\bm{\mathrm{F}}_{\mathrm{c}},\bm{\mathrm{W}}_{{\mathrm{c}},k},\bm{\theta })\in \mathbb{C}^{d \times d}$ is the mean-square error (MSE) matrix of UE $k$ given by 
\begin{equation}
    \label{equ36}
    \begin{aligned}
        \nonumber
    &\bm{\mathrm{E}}_k(\bm{\mathrm{F}}_{\mathrm{c}},\bm{\mathrm{W}}_{\mathrm{c},k},\bm{\mathrm{\theta}})
    = \mathbb{E} [( \widehat{\bm{\mathrm{c}}}  _{k}- \bm{\mathrm{c}}_k)(\widehat{\bm{\mathrm{c}}}  _{k} -\bm{\mathrm{c}}_k )^{\mathrm{H}}]\\
    &= (\bm{\mathrm{W}}_{\mathrm{c},k}^{\mathrm{H}} \bm{\mathrm{H}}_k \bm{\mathrm{F}}_{\mathrm{c},k}  -\bm{\mathrm{I}})
    (\bm{\mathrm{W}}_{\mathrm{c},k}^{\mathrm{H}} \bm{\mathrm{H}}_k\bm{\mathrm{F}}_{\mathrm{c},k} -\bm{\mathrm{I}})^{\mathrm{H}} \\
    &+ \sum ^K _{ i = 1,i\neq k}\bm{\mathrm{W}}_{\mathrm{c},k}^{\mathrm{H}} \bm{\mathrm{H}}_i \bm{\mathrm{F}}_{\mathrm{c},i} \bm{\mathrm{F}}_{\mathrm{c},i} ^{\mathrm{H}}\bm{\mathrm{H}}_i^{\mathrm{H}} \bm{\mathrm{W}}_{\mathrm{c},k}
    + \sigma _c^2 \bm{\mathrm{W}}_{\mathrm{c},k}^{\mathrm{H}} \bm{\mathrm{W}}_{\mathrm{c},k}. 
    \end{aligned}
\end{equation}

In the following, we design an algorithm for solving Problem $\mathcal{P} 2\text{-}4$ by alternately optimizing the beamforming variables $\{\mathbf{F}_{\mathrm{c}}, \mathbf{W}_{\mathrm{c}},\mathbf{w}_{\mathrm{s}},\bm{\theta}\}$ under the BCD framework. Given $\{\bm{\mathrm{F}}_{\mathrm{c}},\bm{\mathrm{w}}_{\mathrm{s}},\bm{\theta}\}$, Problem $\mathcal{P} 2\text{-}4$ is mariginal convex w.r.t. $\{\bm{\mathrm{W}}_{\mathrm{c}},\bm{\mathrm{D}}\}$. Hence, by setting the first-order derivative of the OF of Problem $\mathcal{P} 2\text{-}4$ w.r.t. $\bm{\mathrm{W}}_{{\mathrm{c}},k}$ and $\bm{\mathrm{D}}_k$ to zero, respectively, the optimal solutions are obtained as
\begin{align}
    &\bm{\mathrm{W}}_{{\mathrm{c}},k}^{\star} = (\bm{\mathrm{J}}_k + \bm{\mathrm{H}}_k \bm{\mathrm{F}}_{{\mathrm{c}},k}\bm{\mathrm{F}}_{{\mathrm{c}},k}^{\mathrm{H}}\bm{\mathrm{H}}_k^{\mathrm{H}})^{-1}\bm{\mathrm{H}}_k \bm{\mathrm{F}}_{{\mathrm{c}},k}, \label{equw}\\
    &\bm{\mathrm{D}}_k^{\star} = (\bm{\mathrm{E}}_k^{{\mathrm{(MMSE)}}})^{-1}, \label{equw2}
\end{align}
where 
\begin{equation}
    \bm{\mathrm{E}}_k^{\mathrm{(MMSE)}}= \bm{\mathrm{I}}_d - \bm{\mathrm{F}}_{{\mathrm{c}},k}^{\mathrm{H}} \bm{\mathrm{H}}_k^{\mathrm{H}} (\bm{\mathrm{J}}_k + \bm{\mathrm{H}}_k \bm{\mathrm{F}}_{{\mathrm{c}},k}\bm{\mathrm{F}}_{{\mathrm{c}},k}^{\mathrm{H}}\bm{\mathrm{H}}_k^{\mathrm{H}})^{-1}\bm{\mathrm{H}}_k \bm{\mathrm{F}}_{{\mathrm{c}},k} \nonumber
\end{equation}
is the minimum mean square error (MMSE) of UE $k$.

Meanwhile, $R_k$ is given by
\begin{equation}
    \label{equr}
    R_k = B_{\omega } \log_2 |(\bm{\mathrm{E}}_k^{(MMSE)})^{-1}|.
\end{equation}

\paragraph{Receive and Transmit Beamforming Design for the UE}We first consider the optimization of the receive beamforming vector $\bm{\mathrm{w}}_{{\mathrm{s}},k}$ of UE $k$ with fixed $\bm{\mathrm{F}}_{\mathrm{c}}$. Note that $\bm{\mathrm{w}}_{{\mathrm{s}},k}$ is only related to the radar SINR constraint (\ref{eqb}). The optimal $\bm{\mathrm{w}}_{{\mathrm{s}},k}$ that maximizes the sensing SINR is given by 
\begin{equation}
    \label{equ12}
        \begin{aligned}
            \bm{\mathrm{w}}_{\mathrm{s},k}^{\star} 
            &= \arg \max _{\bm{\mathrm{w}}_{\mathrm{s},k}}\frac{\bm{\mathrm{w}}_{\mathrm{s},k}^{\mathrm{H}} \bm{\mathrm{G}}_k \bm{\mathrm{F}}_{\mathrm{c},k} \bm{\mathrm{F}}_{\mathrm{c},k} ^{\mathrm{H}}\bm{\mathrm{G}}_k^{\mathrm{H}} \bm{\mathrm{w}}_{\mathrm{s},k}}{\bm{\mathrm{w}}_{\mathrm{s},k}^{\mathrm{H}} \bm{\mathrm{T}}_k\bm{\mathrm{w}}_{\mathrm{s},k}},
        \end{aligned}
\end{equation}
which equals the eigenvector corresponding to the largest generalized eigenvalue of the matrix $[\bm{\mathrm{T}}_k^{-1} \bm{\mathrm{G}}_k \bm{\mathrm{F}}_{{\mathrm{c}},k} \bm{\mathrm{F}}_{{\mathrm{c}},k} ^{\mathrm{H}}\bm{\mathrm{G}}_k^{\mathrm{H}}]$ according to the MVDR method \cite{36}.

In the following, we focus on optimizing the transmit beamforming matrices $\bm{\mathrm{F}}_{\mathrm{c}}$ with the fixed $\bm{\mathrm{w}}_{{\mathrm{s}},k}$.  First, the non-convex fractional inequalities (\ref{eqb}) can be rewritten as
\begin{small}
    \begin{align}
        \label{equ39}
        \eta \mathrm{Tr}(\bm{\mathrm{T}}_k\bm{\mathrm{w}}_{{\mathrm{s}},k}\bm{\mathrm{w}}_{{\mathrm{s}},k}^{\mathrm{H}}) - \mathrm{Tr}(\bm{\mathrm{G}}_k \bm{\mathrm{F}}_{{\mathrm{c}},k} \bm{\mathrm{F}}_{{\mathrm{c}},k} ^{\mathrm{H}}\bm{\mathrm{G}}_k^{\mathrm{H}} \bm{\mathrm{w}}_{{\mathrm{s}},k}\bm{\mathrm{w}}_{{\mathrm{s}},k}^{\mathrm{H}})\leqslant 0, \forall k \in \mathcal{K},
    \end{align}
\end{small}%
where the SCA method can be applied to derive a convex upper bound of (\ref{equ39}). 
By defining $f(\bm{\mathrm{F}}_{\mathrm{c}}) \triangleq   -\mathrm{Tr}(\bm{\mathrm{G}}_k \bm{\mathrm{F}}_{\mathrm{c},k} \bm{\mathrm{F}}_{\mathrm{c},k} ^{\mathrm{H}}\bm{\mathrm{G}}_k^{\mathrm{H}} \bm{\mathrm{w}}_{\mathrm{s},k}\bm{\mathrm{w}}_{\mathrm{s},k}^{\mathrm{H}})$, the first-order Taylor approximation of $f(\bm{\mathrm{F}}_{\mathrm{c}})$ at $\bm{\mathrm{F}}_{\mathrm{c}}^{(t)}$ in the $t$-th iteration can be formulated as 
\begin{equation}
    \label{equFun}
        \begin{aligned}
        f(\bm{\mathrm{F}}_{\mathrm{c}},\bm{\mathrm{F}}_{\mathrm{c}}^{(t)}) & = \mathrm{Tr}(\bm{\mathrm{G}}_k^{\mathrm{T}} \bm{\mathrm{w}}_{\mathrm{s},k}^* \bm{\mathrm{w}}_{\mathrm{s},k}^{\mathrm{T}} \bm{\mathrm{G}}_k^* \bm{\mathrm{F}}_{\mathrm{c},k}^{(t),*} \bm{\mathrm{F}}_{\mathrm{c},k}^{(t),\mathrm{T}} ) \\
        &  - 2\mathrm{Re} \{ \mathrm{Tr}(\bm{\mathrm{G}}_k^{\mathrm{H}}  \bm{\mathrm{w}}_{\mathrm{s},k}\bm{\mathrm{w}}_{\mathrm{s},k}^{\mathrm{H}} \bm{\mathrm{G}}_k\bm{\mathrm{F}}_{\mathrm{c},k} \bm{\mathrm{F}}_{\mathrm{c},k} ^{(t),\mathrm{H}}) \}.
        \end{aligned}
\end{equation}
Then, with $\{\bm{\mathrm{W}}_{\mathrm{c}},\bm{\mathrm{w}}_{\mathrm{s}}, \bm{\mathrm{\theta}},\bm{\mathrm{D}}\}$ fixed, Problem $\mathcal{P} 2\text{-}4$ can be approximated as
\begin{small}
\label{equx}
    \begin{align}
        \mathcal{P} 2\text{-}5:\min_{\bm{\mathrm{F}}_{\mathrm{c}}} \quad &\sum _{k=1}^K \sum _{i=1}^K \delta _k \lambda_k  \mathrm{Tr}(\bm{\mathrm{D}}_k \bm{\mathrm{W}}_{\mathrm{c},k}^{\mathrm{H}} \bm{\mathrm{H}}_i \bm{\mathrm{F}}_{\mathrm{c},i}\bm{\mathrm{F}}_{\mathrm{c},i}^{\mathrm{H}} \bm{\mathrm{H}}_i^{\mathrm{H}} \bm{\mathrm{W}}_{\mathrm{c},k} )  - \nonumber \\ 
        &\sum _{k=1}^K \delta _k \lambda_k \mathrm{Tr}(\bm{\mathrm{D}}_k( \bm{\mathrm{F}}_{\mathrm{c},k}^{\mathrm{H}} \bm{\mathrm{H}}_k^{\mathrm{H}} \bm{\mathrm{W}}_{\mathrm{c},k} + \bm{\mathrm{W}}_{\mathrm{c},k}^{\mathrm{H}} \bm{\mathrm{H}}_k \bm{\mathrm{F}}_{\mathrm{c},k} )) \nonumber \\
        s.t.\quad 
        & \mathrm{Tr}(\bm{\mathrm{R}}_{\bm{\mathrm{X}},k}) \leqslant P_{t,k},\forall k \in \mathcal{K},\\
        & \eta \mathrm{Tr}(\bm{\mathrm{T}}_k\bm{\mathrm{w}}_{\mathrm{s},k}\bm{\mathrm{w}}_{\mathrm{s},k}^{\mathrm{H}}) + f(\bm{\mathrm{F}}_{\mathrm{c}},\bm{\mathrm{F}}_{\mathrm{c}}^{(t)}) \leqslant 0, \forall k \in \mathcal{K}. \nonumber 
    \end{align}
\end{small}%
It can be verified that Problem $\mathcal{P} 2\text{-}5$ is a convex problem, which can be efficiently solved via CVX\cite{cvx}.

\paragraph{RIS Phase Shift Design}
Finally, we focus on the optimization of the phase shift of the RIS $\bm{\theta}$. The corresponding subproblem can be formulated as
\begin{small}
    \label{equ40}
    \begin{align}
        \mathcal{P} 2\text{-}6:\min_{\bm{\mathrm{\theta}}}\quad & \sum _{k=1}^K \sum _{i=1}^K \delta _k \lambda_k  \mathrm{Tr}(\bm{\mathrm{D}}_k \bm{\mathrm{W}}_{\mathrm{c},k}^{\mathrm{H}} \bm{\mathrm{H}}_i \bm{\mathrm{F}}_{\mathrm{c},i}\bm{\mathrm{F}}_{\mathrm{c},i}^{\mathrm{H}} \bm{\mathrm{H}}_i^{\mathrm{H}} \bm{\mathrm{W}}_{\mathrm{c},k} )  - \nonumber \\ 
        &\sum _{k=1}^K \delta _k \lambda_k \mathrm{Tr}(\bm{\mathrm{D}}_k( \bm{\mathrm{F}}_{\mathrm{c},k}^{\mathrm{H}} \bm{\mathrm{H}}_k^{\mathrm{H}} \bm{\mathrm{W}}_{\mathrm{c},k} + \bm{\mathrm{W}}_{\mathrm{c},k}^{\mathrm{H}} \bm{\mathrm{H}}_k \bm{\mathrm{F}}_{\mathrm{c},k} )) \nonumber \\ 
        s.t. \quad & 0 < \theta _l \leqslant 2\pi , \forall l \in \mathcal{L}.
    \end{align}
\end{small}
By defining $\bm{\mathrm{H}}_k \triangleq \bm{\mathrm{H}}_{\mathrm{bu},k}+\bm{\mathrm{H}}_{\mathrm{r}} \bm{\mathrm{V}} \bm{\mathrm{H}}_{\mathrm{ru},k}$, we have 
\begin{equation}
    \label{equ44}
    \begin{aligned}
        &\bm{\mathrm{D}}_k \bm{\mathrm{W}}_{\mathrm{c},k}^{\mathrm{H}} \bm{\mathrm{H}}_i \bm{\mathrm{F}}_{\mathrm{c},i}\bm{\mathrm{F}}_{\mathrm{c},i}^{\mathrm{H}} \bm{\mathrm{H}}_i^{\mathrm{H}} \bm{\mathrm{W}}_{\mathrm{c},k}\\
        &=\bm{\mathrm{D}}_k \bm{\mathrm{W}}_{\mathrm{c},k}^{\mathrm{H}} \bm{\mathrm{H}}_{\mathrm{bu},i} \bm{\mathrm{F}}_{\mathrm{c},i}\bm{\mathrm{F}}_{\mathrm{c},i}^{\mathrm{H}}  \bm{\mathrm{H}}_{\mathrm{ru},i} ^{\mathrm{H}} \bm{\mathrm{V}}^{\mathrm{H}} \bm{\mathrm{H}}_{\mathrm{r}}^{\mathrm{H}} \bm{\mathrm{W}}_{\mathrm{c},k} \\
        &+ \bm{\mathrm{D}}_k \bm{\mathrm{W}}_{\mathrm{c},k}^{\mathrm{H}} \bm{\mathrm{H}}_{\mathrm{r}} \bm{\mathrm{V}} \bm{\mathrm{H}}_{\mathrm{ru},i} \bm{\mathrm{F}}_{\mathrm{c},i}\bm{\mathrm{F}}_{\mathrm{c},i}^{\mathrm{H}}\bm{\mathrm{H}}_{\mathrm{bu},i}^{\mathrm{H}}\bm{\mathrm{W}}_{\mathrm{c},k} \\
        &+\bm{\mathrm{D}}_k \bm{\mathrm{W}}_{\mathrm{c},k}^{\mathrm{H}} \bm{\mathrm{H}}_{\mathrm{r}} \bm{\mathrm{V}} \bm{\mathrm{H}}_{\mathrm{ru},i} \bm{\mathrm{F}}_{\mathrm{c},i}\bm{\mathrm{F}}_{\mathrm{c},i}^{\mathrm{H}} \bm{\mathrm{H}}_{\mathrm{ru},i} ^{\mathrm{H}} \bm{\mathrm{V}}^{\mathrm{H}}\bm{\mathrm{H}}_{\mathrm{r}}^{\mathrm{H}} \bm{\mathrm{W}}_{\mathrm{c},k} \\
        &+ \bm{\mathrm{D}}_k \bm{\mathrm{W}}_{\mathrm{c},k}^{\mathrm{H}} \bm{\mathrm{H}}_{\mathrm{bu},i} \bm{\mathrm{F}}_{\mathrm{c},i}\bm{\mathrm{F}}_{\mathrm{c},i}^{\mathrm{H}} \bm{\mathrm{H}}_{\mathrm{bu},i}^{\mathrm{H}}\bm{\mathrm{W}}_{\mathrm{c},k}, 
    \end{aligned}
\end{equation}
and 
\begin{equation}
    \label{equ45}
    \begin{aligned}
        &\bm{\mathrm{D}}_k( \bm{\mathrm{F}}_{\mathrm{c},k}^{\mathrm{H}} \bm{\mathrm{H}}_k^{\mathrm{H}} \bm{\mathrm{W}}_{\mathrm{c},k} + \bm{\mathrm{W}}_{\mathrm{c},k}^{\mathrm{H}} \bm{\mathrm{H}}_k \bm{\mathrm{F}}_{\mathrm{c},k}) \\
        & = \bm{\mathrm{D}}_k\bm{\mathrm{F}}_{\mathrm{c},k}^{\mathrm{H}}\bm{\mathrm{H}}_{\mathrm{ru},k} ^{\mathrm{H}} \bm{\mathrm{V}}^{\mathrm{H}} \bm{\mathrm{H}}_{\mathrm{r}}^{\mathrm{H}} \bm{\mathrm{W}}_{\mathrm{c},k} + \bm{\mathrm{D}}_k\bm{\mathrm{F}}_{\mathrm{c},k}^{\mathrm{H}} \bm{\mathrm{H}}_{\mathrm{bu},k}^{\mathrm{H}} \bm{\mathrm{W}}_{\mathrm{c},k} \\
        & + \bm{\mathrm{D}}_k \bm{\mathrm{W}}_{\mathrm{c},k}^{\mathrm{H}} \bm{\mathrm{H}}_{\mathrm{r}} \bm{\mathrm{V}} \bm{\mathrm{H}}_{\mathrm{ru},k} \bm{\mathrm{F}}_{\mathrm{c},k} +\bm{\mathrm{D}}_k \bm{\mathrm{W}}_{\mathrm{c},k}^{\mathrm{H}} \bm{\mathrm{H}}_{\mathrm{bu},k} \bm{\mathrm{F}}_{\mathrm{c},k}. 
    \end{aligned}
\end{equation}
By defining $\bm{\mathrm{S}}_{k,i} \triangleq \bm{\mathrm{H}}_{\mathrm{r}}^{\mathrm{H}} \bm{\mathrm{W}}_{\mathrm{c},k} \bm{\mathrm{D}}_k \bm{\mathrm{W}}_{\mathrm{c},k}^{\mathrm{H}} \bm{\mathrm{H}}_{\mathrm{bu},i} \bm{\mathrm{F}}_{\mathrm{c},i}\bm{\mathrm{F}}_{\mathrm{c},i}^{\mathrm{H}} \bm{\mathrm{H}}_{\mathrm{ru},i} ^{\mathrm{H}} $, $\bm{\mathrm{B}}_i \triangleq \bm{\mathrm{H}}_{\mathrm{ru},i} \bm{\mathrm{F}}_{\mathrm{c},i}\bm{\mathrm{F}}_{\mathrm{c},i}^{\mathrm{H}} \bm{\mathrm{H}}_{\mathrm{ru},i} ^{\mathrm{H}} $,  $\bm{\mathrm{C}}_k \triangleq \bm{\mathrm{H}}_{\mathrm{r}}^{\mathrm{H}} \bm{\mathrm{W}}_{\mathrm{c},k}\bm{\mathrm{D}}_k \bm{\mathrm{W}}_{\mathrm{c},k}^{\mathrm{H}} \bm{\mathrm{H}}_{\mathrm{r}}$, and $\bm{\mathrm{P}}_{k,i} \triangleq  \bm{\mathrm{H}}_{\mathrm{ru},i} \bm{\mathrm{F}}_{\mathrm{c},k} \bm{\mathrm{D}}_k \bm{\mathrm{W}}_{\mathrm{c},k}^{\mathrm{H}} \bm{\mathrm{H}}_{\mathrm{r}}$, (\ref{equ44}) and (\ref{equ45}) are further transformed to
\begin{align}
    &\mathrm{Tr}(\bm{\mathrm{D}}_k \bm{\mathrm{W}}_{\mathrm{c},k}^{\mathrm{H}} \bm{\mathrm{H}}_i \bm{\mathrm{F}}_{\mathrm{c},i}\bm{\mathrm{F}}_{\mathrm{c},i}^{\mathrm{H}} \bm{\mathrm{H}}_i^{\mathrm{H}} \bm{\mathrm{W}}_{\mathrm{c},k} ) \label{equ46}\\
    &=\mathrm{Tr}(\bm{\mathrm{V}} \bm{\mathrm{B}}_i \bm{\mathrm{V}} ^{\mathrm{H}} \bm{\mathrm{C}}_k)+\mathrm{Tr}(\bm{\mathrm{V}} ^{\mathrm{H}}\bm{\mathrm{S}}_{k,i}) + \mathrm{Tr}( \bm{\mathrm{V}} \bm{\mathrm{S}}_{k,i}^{\mathrm{H}})+\mathrm{const}_1,\nonumber
\end{align}
and
\begin{equation}
    \label{equ47}
\begin{aligned}
     &\mathrm{Tr}(\bm{\mathrm{D}}_k( \bm{\mathrm{F}}_{\mathrm{c},k}^{\mathrm{H}} \bm{\mathrm{H}}_k^{\mathrm{H}} \bm{\mathrm{W}}_{\mathrm{c},k} + \bm{\mathrm{W}}_{\mathrm{c},k}^{\mathrm{H}} \bm{\mathrm{H}}_k \bm{\mathrm{F}}_{\mathrm{c},k})) \\
     &= \mathrm{Tr}(\bm{\mathrm{P}}_{k,i}\bm{\mathrm{V}}) + \mathrm{Tr}( \bm{\mathrm{P}}_{k,i}^{\mathrm{H}} \bm{\mathrm{V}}^{\mathrm{H}})+\mathrm{const}_2,
\end{aligned}
\end{equation}
where $\mathrm{const}_1$ and $\mathrm{const}_2$ are constant independent of $\bm{\mathrm{V}}$. 
Thus, Problem $\mathcal{P} 2\text{-}6$ can be rewritten as 
    \begin{align}
        \label{equ41}
        \mathcal{P} 2\text{-}7:\min_{\bm{\mathrm{\theta}}} \quad&\mathrm{Tr}(\bm{\mathrm{V}} \bm{\mathrm{B}} \bm{\mathrm{V}} ^{\mathrm{H}} \bm{\mathrm{C}}) +
        \mathrm{Tr}(\bm{\mathrm{V}}\bm{\mathrm{U}}) +  \mathrm{Tr}(\bm{\mathrm{V}}^{\mathrm{H}} \bm{\mathrm{U}}^{\mathrm{H}}) \\
        s.t. \quad& 0 < \theta _l \leqslant 2\pi ,\forall l \in \mathcal{L}, \nonumber
    \end{align}
where  $\bm{\mathrm{B}} \triangleq  \sum _{i=1}^K\bm{\mathrm{B}}_i $, $\bm{\mathrm{C}} \triangleq  \sum _{k=1}^K \delta _k \lambda_k\bm{\mathrm{C}}_k$, and $\bm{\mathrm{U}} \triangleq  \sum _{i=1}^K\sum _{k=1}^K\delta _k \lambda_k (\bm{\mathrm{S}}_{k,i}^{\mathrm{H}}-\bm{\mathrm{P}}_{k,i})$.

With the definitions of $ \bm{\mathrm{\phi }} \triangleq  [e^{j\theta _1},\cdots,e^{j\theta _L} ]^{\mathrm{T}}$ and $\bm{\mathrm{u}} = [[\bm{\mathrm{U}}]_{1,1},\cdots,[\bm{\mathrm{U}}]_{L,L}]^{\mathrm{T}}$, we have 
\begin{equation}
    \mathrm{Tr}  (\bm{\mathrm{V}} \bm{\mathrm{B}} \bm{\mathrm{V}} ^{\mathrm{H}} \bm{\mathrm{C}}) = \bm{\mathrm{\phi }} ^{\mathrm{H}} (\bm{\mathrm{C}} \odot  \bm{\mathrm{B}}^{\mathrm{T}})\bm{\mathrm{\phi }}, 
\end{equation}
and 
\begin{equation}
    \mathrm{Tr}(\bm{\mathrm{V}}\bm{\mathrm{U}}) =\bm{\mathrm{\phi }} ^{\mathrm{T}} \bm{\mathrm{u}}, \mathrm{Tr}(\bm{\mathrm{V}}^{\mathrm{H}} \bm{\mathrm{U}}^{\mathrm{H}}) =\bm{\mathrm{u}}^{\mathrm{H}} \bm{\mathrm{\phi }} ^{*}.
\end{equation}
Accordingly, Problem $\mathcal{P} 2\text{-}7$ can be reformulated as
\begin{equation}
    \label{equ43}
    \begin{aligned}
        \mathcal{P} 2\text{-}8:\min_{\bm{\mathrm{\phi}}} \quad & g(\bm{\mathrm{\phi }}) = \bm{\mathrm{\phi }} ^{\mathrm{H}} \bm{\mathrm{\varXi }} \bm{\mathrm{\phi }} + 2\mathrm{Re} \{\bm{\mathrm{\phi }} ^{\mathrm{H}}  \bm{\mathrm{u}}^{*}\} \\
        s.t. \quad & \left\lvert \phi _l\right\rvert = 1,\forall l \in \mathcal{L},
    \end{aligned}
\end{equation}
where $\bm{\mathrm{\varXi }} \triangleq \bm{\mathrm{C}} \odot  \bm{\mathrm{B}}^{\mathrm{T}}$. Problem $\mathcal{P} 2\text{-}8$ is a non-convex problem owing to the unit modulus constraint. In the following, the MM algorithm \cite{55} with two steps is applied to tackle this problem. Specifically, the surrogate function $h(\bm{\mathrm{\phi }}|\bm{\mathrm{\phi }}^{(t)})$ at the $t$-th iteration $\bm{\mathrm{\phi }}^{(t)}$, is introduced in the majorization step to provide an upper bound for $g(\bm{\mathrm{\phi }})$.
Then, in the minimization step, we update $\bm{\mathrm{\phi }}$ by $\bm{\mathrm{\phi }}^{(t+1)} \in \arg \min _{\bm{\mathrm{\phi }}} h(\bm{\mathrm{\phi }}|\bm{\mathrm{\phi }}^{(t)})$. To this end, we introduce the lemma below to construct the surrogate function of $g(\bm{\mathrm{\phi }})$.
\begin{lemma}
    \label{lemma1}
    By denoting the maximum eigenvalue of $\bm{\mathrm{\varXi }}$ by $\lambda _{\max}$, we have
    \begin{equation}
        \begin{aligned}
            &g(\bm{\mathrm{\phi }})\leqslant \bm{\mathrm{\phi }}^{\mathrm{H}} \lambda _{\max} \bm{\mathrm{I}}_L\bm{\mathrm{\phi }} - 
            2\mathrm{Re}\{\bm{\mathrm{\phi }}^{\mathrm{H}}(\lambda _{\max} \bm{\mathrm{I}}_L-\bm{\mathrm{\varXi }} )\bm{\mathrm{\phi }}^{(t)}\}\\
            &+\bm{\mathrm{\phi }}^{(t),{\mathrm{H}}} (\lambda _{\max} \bm{\mathrm{I}}_L-\bm{\mathrm{\varXi }} )\bm{\mathrm{\phi }}^{(t)} 
            +2\mathrm{Re} \{\bm{\mathrm{\phi }} ^{\mathrm{H}}  \bm{\mathrm{u}}^{*}\} \triangleq  h(\bm{\mathrm{\phi }}|\bm{\mathrm{\phi }}^{(t)}).
        \end{aligned}
        \end{equation}
\end{lemma}
\begin{proof}
    Please refer to \cite{55,46}.
\end{proof}  
Note that $\bm{\mathrm{\phi }}^{\mathrm{H}} \bm{\mathrm{\phi }} = L$, $\bm{\mathrm{\phi }}^{\mathrm{H}} \lambda _{\max} \bm{\mathrm{I}}_L\bm{\mathrm{\phi }} = L\lambda _{\max}$ is a constant. Thus, with constants independent of $\bm{\mathrm{\phi}}$ removed,  Problem $\mathcal{P} 2\text{-}8$ with the surrogate function at the $t+1$-th iteration  can be formulated as  
\begin{equation}
    \label{equ49}
    \begin{aligned}
        \mathcal{P} 2\text{-}9:\max_{\bm{\mathrm{\phi }}}\quad & \mathrm{Re}\{ \bm{\mathrm{\phi }}^{\mathrm{H}}[(\lambda _{\max} \bm{\mathrm{I}}_L-\bm{\mathrm{\varXi }} )\bm{\mathrm{\phi }}^{(t)}-\bm{\mathrm{u}}^{*}]\} \\
        s.t. \quad & \left\lvert \phi  _l\right\rvert = 1 ,\forall l \in \mathcal{L}.
    \end{aligned}
\end{equation}
Thus, the optimal solution for $\bm{\theta}$ t the $(t+1)$-th iteration is given by
\begin{equation}
    \label{equtheta}
    \bm{\theta}^{(t+1)} = \angle  ((\lambda _{\max} \bm{\mathrm{I}}_L-\bm{\mathrm{\varXi }} )\bm{\mathrm{\phi }}^{(t)}-\bm{\mathrm{u}}^{*}).
\end{equation}
According to the preceding analysis, the procedure for solving Problem $\mathcal{P}2\text{-}4$ is summarized in Algorithm \ref{alg:algorithm3}. The complexity of this algorithm is provided as follows. Since the number of variables in Step \ref{stepf} is $KNd$, the complexity of updating $\bm{\mathrm{F}}_{\mathrm{c}}^{(t_3+1)}$ via the SCA method is given by $\mathcal{O} ((KNd)^{3.5})$. In Step \ref{stepw}, the complexities of calculating $\bm{\mathrm{W}}_{\mathrm{c}}^{(t_3+1)}$ and $\bm{\mathrm{D}}^{(t_3+1)}$ are on the order of $\mathcal{O} (KM^3)$ and $\mathcal{O} (Kd^3)$, respectively. In Step \ref{stepws}, the complexity of calculating $\bm{\mathrm{w}}_{\mathrm{s}}^{(t_3+1)}$ is $\mathcal{O} (KN^3)$. Defining the number of iterations as $t_{MM}^{\max}$, the computation complexity of updating $\bm{\mathrm{\theta }}^{(t_3+1)}$ via the MM algorithm in Step \ref{steptheta} contains two parts. Specifically, the complexity of calculating the maximum eigenvalue $\lambda _{\max}$ of $\bm{\mathrm{\varXi }}$ is on the order of $\mathcal{O} (L^3)$. Note that the main complexity for each iteration of the MM algorithm arises from calculating $\bm{\mathrm{\phi }}^{(t)}$, whose complexity is given by $\mathcal{O}(L^2)$. Hence, the total complexity of updating $\bm{\mathrm{\theta }}^{(t_3+1)}$ is on the order of $\mathcal{O} (L^3+t_{MM}^{\max} L^2)$. By aggregating these components, the overall complexity of Algorithm \ref{alg:algorithm3} is given by $\mathcal{O} (t_3^{\max}\cdot\max\{(KNd)^{3.5}, KM^3, Kd^3, KN^3, L^3+t_{MM}^{\max} L^2\})$. Furthermore, the complexity of Algorithm \ref{alg:algorithm2} is dominated by calculating $\{\bm{\mathrm{F}}_{\mathrm{c}}^{(t_2+1)},\bm{\mathrm{W}}_{\mathrm{c}}^{(t_2+1)},\bm{\mathrm{w}}_{\mathrm{s}}^{(t_2+1)}, \bm{\mathrm{\theta }}^{(t_2+1)}\}$ using Algorithm \ref{alg:algorithm3}, as all other steps involve explicit closed-form solutions. 
\begin{algorithm}[t]
    \caption{Joint optimization of $\{\bm{\mathrm{F}}_{\mathrm{c}},\bm{\mathrm{W}}_{\mathrm{c}},\bm{\mathrm{w}}_{\mathrm{s}},\bm{\mathrm{\theta }}\}$, Given $\{\bm{\mathrm{\delta}},\bm{\mathrm{\lambda }}\}$.}
    \label{alg:algorithm3}
    \begin{algorithmic}[1]
      \State Initialize iteration number $t_3 = 0$. Set the maximum number of iterations to $t_3^{\max}$. Initialize feasible $\bm{\mathrm{F}}_{\mathrm{c}}^{(0)}$, $\bm{\mathrm{w}}_{\mathrm{s}}^{(0)}$ and $\bm{\mathrm{\theta }}^{(0)}$ following constraints (\ref{eqa}) - (\ref{eqd}). Calculate $R_k^{(0)}$ using (\ref{equr}).
      \State  \textbf{Repeat} 
      \State \quad Update $\bm{\mathrm{F}}_{\mathrm{c}}^{(t_3+1)}$ by solving Problem $2 \text{-}5$ with the SCA method. \label{stepf}
      \State \quad Calculate $\bm{\mathrm{W}}_{\mathrm{c}}^{(t_3+1)}$ and $\bm{\mathrm{D}}^{(t_3+1)}$ according to (\ref{equw}) and (\ref{equw2}), respectively. \label{stepw}
      \State \quad Calculate $\bm{\mathrm{w}}_{\mathrm{s}}^{(t_3+1)}$ by using (\ref{equ12}). \label{stepws}
      \State \quad Update $\bm{\mathrm{\theta }}^{(t_3+1)}$ using MM algorithm. \label{steptheta}
      \State \quad $t_3 \gets t_3 + 1$.
      \State  \textbf{Until} $t_3=t_3^{\rm{max}}$ \textbf{or}
      \begin{small}
        \begin{align}
            \epsilon _3^{(t_1+1)} = \frac{\left\lvert \mathrm{obj}^{(t_3+1)}(\bm{\mathrm{F}}_{\mathrm{c}},\bm{\mathrm{\theta }}) - \mathrm{obj}^{(t_3)}(\bm{\mathrm{F}}_{\mathrm{c}},\bm{\mathrm{\theta }}) \right\rvert }{\mathrm{obj}^{(t_3+1)}(\bm{\mathrm{F}}_{\mathrm{c}},\bm{\mathrm{\theta }})} < \epsilon.\nonumber
        \end{align}
      \end{small}
    \end{algorithmic}
  \end{algorithm}
 
\subsection{Summarize of the Algorithm Development}
Following the above discussions, we present the overall BCD algorithm for solving Problem $\mathcal{P} 0$ in Algorithm \ref{alg:algorithm4}. The monotonic increase in the OF value of Problem $\mathcal{P} 0$ at each step of Algorithm \ref{alg:algorithm4} is readily verified. Furthermore, the edge computing resource constraints impose an upper bound on the OF value. Hence, Algorithm \ref{alg:algorithm4} is guaranteed to converge. The complexity of Algorithm \ref{alg:algorithm4} is mainly determined by Step \ref{step1} and Step \ref{step2}, the complexities of which have been thoroughly discussed in the preceding subsections. 
\begin{algorithm}[b]
    \caption{Joint optimization of $\bm{\mathrm{v}}$, $\bm{\mathrm{f}}$, $\bm{\mathrm{F}}_{\mathrm{c}}$ and $\bm{\mathrm{\theta }}$.}
    \label{alg:algorithm4}
    \begin{algorithmic}[1]
      \State Initialize iteration number $t_4 = 0$. Set the maximum number of iterations to $t_4^{\max}$. Initialize feasible $\bm{\mathrm{f}}^{(0)}$, $\bm{\mathrm{F}}_{\mathrm{c}}^{(0)}$, $\bm{\mathrm{w}}_{\mathrm{s}}^{(0)}$ and $ \bm{\mathrm{\theta }}^{(0)}$ following constraints (\ref{eqa}) - (\ref{eqc}) and (\ref{eqf}).
      \State  \textbf{Repeat} 
      \State \quad Calculate $\bm{\mathrm{v}}^{(t_4+1)}$ and $\bm{\mathrm{f}}^{(t_4+1)}$ via Algorithm \ref{alg:algorithm1}. \label{step1}
      \State \quad Calculate $\bm{\mathrm{F}}_{\mathrm{c}}^{(t_4+1)}$, $\bm{\mathrm{W}}_{\mathrm{c}}^{(t_4+1)}$, $\bm{\mathrm{w}}_{\mathrm{s}}^{(t_4+1)}$ and $\bm{\mathrm{\theta }}^{(t_4+1)}$ via Algorithm \ref{alg:algorithm2}.\label{step2}
      \State \quad $t_4 \gets t_4 + 1$.
      \State  \textbf{Until} $t_3=t_3^{\rm{max}}$ \textbf{or}
      \begin{small}
        \begin{align}
            \epsilon _4^{(t_1+1)} = \frac{\left\lvert \mathrm{obj}^{(t_4+1)}(\bm{\mathrm{v}},\bm{\mathrm{f}},\bm{\mathrm{F}}_{\mathrm{c}},\bm{\mathrm{\theta }}) - \mathrm{obj}^{(t_4)}(\bm{\mathrm{v}},\bm{\mathrm{f}},\bm{\mathrm{F}}_{\mathrm{c}},\bm{\mathrm{\theta }}) \right\rvert }{\mathrm{obj}^{(t_4+1)}(\bm{\mathrm{v}},\bm{\mathrm{f}},\bm{\mathrm{F}}_{\mathrm{c}},\bm{\mathrm{\theta }})} < \epsilon.\nonumber
        \end{align}
      \end{small}
    \end{algorithmic}
  \end{algorithm}  
\section{Specific Case Study}
\hyphenation{sen-sing}
A special case with a single UE that transmits a single data stream is studied in this section to provide a comprehensive characterization of the RIS-assisted ICSC system. Since the number of the data streams $d$ is assumed to be 1, the radar SINR constraint (\ref{eqb}) degenerates to the SNR constraint. Furthermore, as there is no need for the consideration of edge computing resource allocation, the sum-of-ratios form in the OF of Problem $\mathcal{P} 2$ is transformed into a single-ratio form. In this instance, Problem $\mathcal{P}0$ is simplified as follows:
\begin{subequations}
\label{equ56}
    \begin{align}
        \mathcal{P}3:\min_{\bm{\mathrm{f}}_{\mathrm{c}},\bm{\mathrm{w}}_{\mathrm{c}},\bm{\mathrm{w}}_{\mathrm{s}},\bm{\mathrm{\theta }},v} & T(\bm{\mathrm{f}}_{\mathrm{c}},\bm{\mathrm{w}}_{\mathrm{c}},\bm{\mathrm{\theta}},v)  \\
        s.t. \quad \ 
        & \mathrm{Tr}(\bm{\mathrm{R}}_{\bm{\mathrm{X}}}) \leqslant P_{t}, \label{aaa} \\
        & \gamma \geqslant  \eta , \label{bbb}\ \\
        & 0 < \theta _l \leqslant 2\pi ,\forall l \in \mathcal{L}, \label{ccc}\\
        & v  \in \{0,1,\cdots,V \} . \label{ddd} 
    \end{align}
\end{subequations}
where $\eta$ represents the SNR threshold of the UE for radar sensing, $\bm{\mathrm{f}}_{\mathrm{c}} \in \mathbb{C} ^{N \times 1} $ and $\bm{\mathrm{w}}_{c}\in \mathbb{C} ^{M \times 1}$ denote the precoding and decoding vectors, respectively.

The BCD method is employed to solve Problem $\mathcal{P}3$. As discussed in Sec. \ref{secv}, given $\{\bm{\mathrm{f}}_{\mathrm{c}},\bm{\mathrm{w}}_{\mathrm{c}},\bm{\mathrm{w}}_{\mathrm{s}},\bm{\mathrm{\theta }}\}$, the optimal solution for $v$ is obtained when $T_{\mathrm{l}} = T_{\mathrm{c}}$. Consequently, the optimal $v$ is equal to
\begin{equation}
    v ^{\star } = \arg \min _{\hat{v }\in \{ \left\lfloor \hat{v}^{\star} \right\rfloor, \left\lceil \hat{v}^{\star} \right\rceil  \}} 
    T(\hat{v}),
\end{equation}
where
\begin{equation}
\label{equ57}
    \hat{v }^{\star} = \frac{V c R f_{\mathrm{total}}^e}{f_{\mathrm{total}}^e f^l + c R(f_{\mathrm{total}}^e + f^l)}.
\end{equation}
Given $v$, Problem $\mathcal{P}3$ can be rewritten as 
\begin{equation}
\label{equ59}
    \begin{aligned}
        \mathcal{P}3\text{-}1:\max_{\bm{\mathrm{f}}_{\mathrm{c}},\bm{\mathrm{w}}_{\mathrm{c}},\bm{\mathrm{w}}_{\mathrm{s}},\bm{\mathrm{\theta }}} &R(\bm{\mathrm{f}}_{\mathrm{c}},\bm{\mathrm{w}}_{\mathrm{c}},\bm{\mathrm{\theta }}) \\
        s.t. \quad 
        & (\ref{aaa}),(\ref{bbb}),(\ref{ccc}).
    \end{aligned}
\end{equation}
Since $R(\bm{\mathrm{f}}_{\mathrm{c}},\bm{\mathrm{w}}_{\mathrm{c}},\bm{\mathrm{\theta }})$ monotonically increases w.r.t. $\frac{ \left\lvert \bm{\mathrm{w}}_{c}^{\mathrm{H}} \bm{\mathrm{H}}\bm{\mathrm{f}}_{c} \right\rvert ^2 }{\sigma_c ^2 \left\lVert \bm{\mathrm{w}}_{c}^{\mathrm{H}} \right\rVert ^2}$,  Problem $\mathcal{P}3\text{-}1$ is equivalent to
\begin{equation}
        \begin{aligned}
            \mathcal{P}3\text{-}2:\max_{\bm{\mathrm{f}}_{\mathrm{c}},\bm{\mathrm{w}}_{\mathrm{c}},\bm{\mathrm{w}}_{\mathrm{s}},\bm{\mathrm{\theta }}} & \frac{ \left\lvert \bm{\mathrm{w}}_{c}^{\mathrm{H}} \bm{\mathrm{H}}\bm{\mathrm{f}}_{c} \right\rvert ^2 }{\sigma_c ^2 \left\lVert \bm{\mathrm{w}}_{c}^{\mathrm{H}} \right\rVert ^2} \\
            s.t. \quad 
            & (\ref{aaa}),(\ref{bbb}),(\ref{ccc}).
        \end{aligned}
\end{equation}
Note that $\bm{\mathrm{w}}_{\mathrm{s}}$ is only related to constraint (\ref{bbb}),  while $\bm{\mathrm{w}}_{\mathrm{c}}$ is related to the OF of Problem $\mathcal{P}3\text{-}2$. Therefore, the optimal solution for the decoding matrices $\bm{\mathrm{w}}_{\mathrm{s}}$ and $\bm{\mathrm{w}}_{\mathrm{c}}$ can be derived based on the maximum ratio combining (MRC) criterion as 
\begin{align}
    &\bm{\mathrm{w}}_{\mathrm{s}}^{\star} = \arg \max\left(\frac{\left\lvert\bm{\mathrm{w}}_{s}^{\mathrm{H}}\bm{\mathrm{G}}\bm{\mathrm{f}}_{c} \right\rvert^2 }{\left\lVert \bm{\mathrm{w}}_{s}^{\mathrm{H}}\right\rVert^2 }\right)=\bm{\mathrm{G}}\bm{\mathrm{f}}_{c},\label{equ63} \\
    &\bm{\mathrm{w}}_{\mathrm{c}}^{\star} = \arg \max\left(\frac{\left\lvert \bm{\mathrm{w}}_{c}^{\mathrm{H}}\bm{\mathrm{H}}\bm{\mathrm{f}}_{c} \right\rvert ^2 }{ \left\lVert \bm{\mathrm{w}}_{c}^{\mathrm{H}} \right\rVert ^2 }\right)= \bm{\mathrm{H}}\bm{\mathrm{f}}_{c}.\label{equ64}
\end{align}
By defining $ \bm{\mathrm{h}} = \frac{\bm{\mathrm{H}}^{\mathrm{H}}\bm{\mathrm{w}}_{c}}{\sigma_c \left\lVert \bm{\mathrm{w}}_{c}^{\mathrm{H}} \right\rVert} \in \mathbb{C}^{N\times 1}$ and $\bm{\mathrm{g}} = \frac{\bm{\mathrm{G}}^{\mathrm{H}}\bm{\mathrm{w}}_{s}}{\sigma_s \left\lVert \bm{\mathrm{w}}_{s}^{\mathrm{H}} \right\rVert} \in \mathbb{C}^{N\times 1}$, the subproblem of optimizing $\bm{\mathrm{f}}_{\mathrm{c}}$ can be formulated as follows:
\begin{subequations}
    \label{equff}
        \begin{align}
            \mathcal{P}3\text{-}3:\max_{\bm{\mathrm{f}}_{\mathrm{c}}} \quad&\left\lvert \bm{\mathrm{h}}^{\mathrm{H}} \bm{\mathrm{f}}_{c} \right\rvert  ^2 \\
            s.t. \quad 
            & \left\lVert \bm{\mathrm{f}}_{c} \right\rVert  ^2 \leqslant P_{t}, \label{transmit} \\
            & \left\lvert \bm{\mathrm{g}}^{\mathrm{H}} \bm{\mathrm{f}}_{c} \right\rvert  ^2 \geqslant  \eta .
        \end{align}
\end{subequations}
The following proposition contributes to addressing Problem $\mathcal{P}3\text{-}3$ \cite{38,39}. 
\begin{proposition}
\label{pro2}
The optimal solution to Problem $\mathcal{P}3\text{-}3$ can be expressed as
\begin{equation}
    \bm{\mathrm{f}}_{\mathrm{c}} = a \bm{\mathrm{h}} + b \bm{\mathrm{g}},
\end{equation}
where $a, \, b \in \mathbb{C}$ are given as follows:
\begin{enumerate}[{Case} 1:]
    \item If $ \eta \leqslant \frac{P_t |\bm{\mathrm{g}}^{\mathrm{H}} \bm{\mathrm{h}}|^2 }{\left\lVert \bm{\mathrm{h}} \right\rVert ^2}$,
    \begin{equation}
        \left\{
        \begin{aligned}
            & |a| = \frac{\sqrt{P_t}}{\left\lVert \bm{\mathrm{h}} \right\rVert} ,\\
            & |b| = 0,
        \end{aligned}
        \right.
    \end{equation}
    with an arbitrary phase shift.
    \label{case1}
    \item If $ \frac{P_t |\bm{\mathrm{g}}^{\mathrm{H}} \bm{\mathrm{h}}|^2 }{\left\lVert \bm{\mathrm{h}} \right\rVert ^2} \leqslant \eta \leqslant P_t \left\lVert \bm{\mathrm{g}} \right\rVert ^2$,
    \begin{equation}
        \left\{
        \begin{aligned}
            & |a| = \sqrt{\frac{P_t \left\lVert \bm{\mathrm{g}} \right\rVert^2 - \eta}{\left\lVert \bm{\mathrm{h}} \right\rVert^2 \left\lVert \bm{\mathrm{g}} \right\rVert^2 - |\bm{\mathrm{h}}^{\mathrm{H}} \bm{\mathrm{g}} |^2 }},\\
            & |b| = \frac{\sqrt{\eta}}{\left\lVert \bm{\mathrm{g}} \right\rVert^2} - \frac{|\bm{\mathrm{h}}^{\mathrm{H}} \bm{\mathrm{g}} |}{\left\lVert \bm{\mathrm{g}} \right\rVert^2} |a|,
        \end{aligned}
        \right.
    \end{equation}
    where the phase of $a$ and $b$ should satisfy
    \begin{equation}
    \angle(a) - \angle(b) = \angle(\bm{\mathrm{h}}^{\mathrm{H}} \bm{\mathrm{g}}).
    \end{equation}
    \label{case2}
    \item If $\eta > P_t \left\lVert \bm{\mathrm{g}} \right\rVert ^2$, there exists no feasible solution. \label{case3}
\end{enumerate}
\end{proposition}
\begin{proof}
    Please refer to Appendix \ref{app2}.
\end{proof}
Given $\{v,\bm{\mathrm{f}}_{\mathrm{c}},\bm{\mathrm{w}}_{\mathrm{s}},\bm{\mathrm{w}}_{\mathrm{c}}\}$, the following inequality holds for the OF of Problem $\mathcal{P}3\text{-}2$.
\begin{small}
\begin{align}
    \frac{|\bm{\mathrm{w}}_{c}^{\mathrm{H}} (\bm{\mathrm{H}}_{\mathrm{bu}}+\bm{\mathrm{H}}_{\mathrm{r}} \bm{\mathrm{V}} \bm{\mathrm{H}}_{\mathrm{ru}})\bm{\mathrm{f}}_{c}|^2}{\sigma_c ^2 \left\lVert\bm{\mathrm{w}}_{c}^{\mathrm{H}} \right\rVert ^2}
    \leqslant  \frac{|\bm{\mathrm{w}}_{c}^{\mathrm{H}} \bm{\mathrm{H}}_{\mathrm{bu}}\bm{\mathrm{f}}_{c}|^2}{\sigma_c ^2 \left\lVert\bm{\mathrm{w}}_{c}^{\mathrm{H}} \right\rVert^2} + \frac{|\bm{\mathrm{w}}_{c}^{\mathrm{H}} \bm{\mathrm{H}}_{\mathrm{r}} \bm{\mathrm{V}} \bm{\mathrm{H}}_{\mathrm{ru}}\bm{\mathrm{f}}_{c}|^2}{\sigma_c ^2 \left\lVert\bm{\mathrm{w}}_{c}^{\mathrm{H}} \right\rVert ^2}.
\end{align}
\end{small}%
Therefore, the optimal solution of $\bm{\mathrm{\theta }}$ can be obtained via the following equation
\begin{equation}
\label{equ69}
    \bm{\mathrm{\theta }} = \angle (\bm{\mathrm{w}}_{c}^{\mathrm{H}} \bm{\mathrm{H}}_{\mathrm{bu}}\bm{\mathrm{f}}_{c}) - \angle (\mathrm{diag}\{\bm{\mathrm{w}}_{c}^{\mathrm{H}} \bm{\mathrm{H}}_{\mathrm{r}}\}\bm{\mathrm{H}}_{\mathrm{ru}}\bm{\mathrm{f}}_{c}).
\end{equation}
The overall algorithm for Problem $\mathcal{P}3$ is summarized in Algorithm \ref{alg:algorithm5}. Next, we analyze the complexity of this algorithm. Note that the complexity of Algorithm \ref{alg:algorithm5} is dominated by Steps 4\,-\,6. Firstly, the complexities of using the MRC criterion to calculate $\bm{\mathrm{w}}_{\mathrm{c}}^{(t_5+1)}$ and $\bm{\mathrm{w}}_{\mathrm{s}}^{(t_5+1)}$ are on the order of $\mathcal{O} (N^2)$ and $\mathcal{O} (MN)$, respectively. Secondly, the complexity of calculating $\bm{\mathrm{f}}_{\mathrm{c}}^{(t_5+1)}$ following Proposition \ref{pro2} is given by $\mathcal{O} (\max\{MN,N^2\})$. Finally, the complexity of calculating $\bm{\mathrm{\theta }}^{(t_5+1)}$ by using (\ref{equ69}) is $\mathcal{O} (L^2N)$. Considering that $L\gg \{M,N\}$, the complexity of Algorithm \ref{alg:algorithm5} is on the order of $\mathcal{O} (L^2N)$.
\begin{algorithm}[t]
    \caption{Joint optimization of $v$, $\bm{\mathrm{w}}_{\mathrm{c}}$, $\bm{\mathrm{w}}_{\mathrm{s}}$ ,$\bm{\mathrm{f}}_{\mathrm{c}}$ and $\bm{\mathrm{\theta }}$}
    \label{alg:algorithm5}
    \begin{algorithmic}[1]
      \State Initialize iteration number $t_5 = 0$. Set the maximum number of iterations to $t_5^{\max}$. Initialize feasible $\bm{\mathrm{f}}_{\mathrm{c}}^{(0)}$, $\bm{\mathrm{w}}_{\mathrm{c}}^{(0)}$, $\bm{\mathrm{w}}_{\mathrm{s}}^{(0)}$ and $ \bm{\mathrm{\theta }}^{(0)}$ following constraints (\ref{aaa}) - (\ref{ddd}).
      \State  \textbf{Repeat} 
      \State \quad Calculate $v$ by using (\ref{equ57})\label{stepv}
      \State \quad Calculate $\bm{\mathrm{w}}_{\mathrm{c}}^{(t_5+1)}$ and $\bm{\mathrm{w}}_{\mathrm{s}}^{(t_5+1)}$ by using (\ref{equ63}) and (\ref{equ64}), respectively.
      \State \quad Calculate $\bm{\mathrm{f}}_{\mathrm{c}}^{(t_5+1)}$ following  Proposition \ref{pro2}.
      \State \quad Calculate $\bm{\mathrm{\theta }}^{(t_5+1)}$ by using (\ref{equ69}).
      \State \quad $t_5 \gets t_5 + 1$.
      \State  \textbf{Until} $t_5=t_5^{\rm{max}}$ \textbf{or}
      \begin{small}
        \begin{align}
            \epsilon _5^{(t_5+1)} = \frac{\left\lvert \mathrm{obj}^{(t_5+1)}(v, \bm{\mathrm{w}}_{\mathrm{c}}, \bm{\mathrm{f}}_{\mathrm{c}}, \bm{\mathrm{\theta }}) - \mathrm{obj}^{(t_5)}(v, \bm{\mathrm{w}}_{\mathrm{c}}, \bm{\mathrm{f}}_{\mathrm{c}}, \bm{\mathrm{\theta }}) \right\rvert }{\mathrm{obj}^{(t_5+1)}(v, \bm{\mathrm{w}}_{\mathrm{c}},\bm{\mathrm{f}}_{\mathrm{c}},\bm{\mathrm{\theta }})} < \epsilon .\nonumber
        \end{align}
      \end{small}
    \end{algorithmic}
\end{algorithm}
\section{Simulation Results}
\hyphenation{co-rresponding}
In this section, numerical results are provided to validate the performance of the proposed RIS-assisted ICSC system in both single-UE and multi-UE scenarios. As shown in Fig. \ref{fig2}, the BS and the RIS are located at (0, 0) and (200 \textrm{m}, 0), respectively. In the multi-UE scenario, we assume that two UEs are located at (240 \textrm{m}, 50 \textrm{m}) and (250 \textrm{m}, -50 \textrm{m}), with corresponding targets at (260 \textrm{m}, 50 \textrm{m}) and (250 \textrm{m}, -70 \textrm{m}), respectively. For the single-UE scenario, the UE is located at (250 \textrm{m}, 50 \textrm{m}) with its target at (270 \textrm{m}, 50 \textrm{m}). 

\begin{figure}
    \centering
    \includegraphics[width=2.3in]{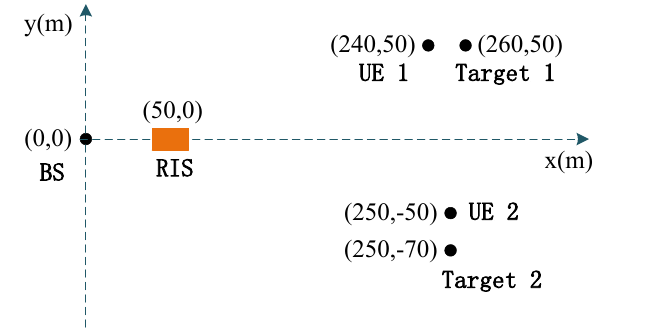}
    \caption{The simulated multi-UE RIS-assisted ICSC scenario}
    \label{fig2}
\end{figure}

\hyphenation{a-ccount}
In the context of wireless communication channels, we account for both small-scale fading and large-scale path loss. In specific, the large-scale path loss in \textrm{dB} is expressed as
\begin{equation}
    \mathrm{PL} = \mathrm{PL}_0 - 10 \alpha \log_{10}  \left(\frac{d}{d_0}\right), 
\end{equation}
where $\mathrm{PL}_0$ is the path loss at the reference distance $d_0$, $d$ represents the link distance and $\alpha$ denotes the corresponding path loss exponent. 
The path loss exponent of the UE-BS link, UE-RIS link,  RIS-BS link, and UE-UE interference channel are denoted by $\alpha_{\mathrm{bu}}$, $\alpha_{\mathrm{ru}}$, and $\alpha_{\mathrm{r}}$, $\alpha_{\mathrm{uu}}$, respectively. The Rayleigh channel model is employed to characterize the communication channels between the BS and UEs, as well as the channels between UEs, due to the presence of potential scatters and obstacles. While the small-scale fading of channels associated with the RIS is assumed to follow a Rician fading distribution:
\begin{equation}
    \bm{\mathrm{\widetilde{H}}} = \sqrt{\frac{K_R}{1+K_R}}\bm{\mathrm{\widetilde{H}}}_{\mathrm{LoS}} +  \sqrt{\frac{1}{1+K_R}}\bm{\mathrm{\widetilde{H}}}_{\mathrm{NLoS}},
\end{equation}
where $K_R$ is the Rician factor, $\bm{\mathrm{\widetilde{H}}}_{\mathrm{LoS}}$ denotes the line of sight (LoS) component, and $\bm{\mathrm{\widetilde{H}}}_{\mathrm{NLoS}}$ represents the non-line of sight (NLoS) component. The LoS component $\bm{\mathrm{\widetilde{H}}}_{\mathrm{LoS}}$ is modeled as $\bm{\mathrm{\widetilde{H}}}_{\mathrm{LoS}} =\bm{\mathrm{a}}_r(\theta^{AoA}) \bm{\mathrm{a}}_t^{\mathrm{H}}(\theta^{AoD}) $, where the steering vectors w.r.t. the angle of arrival (AoA) $\theta^{AoA}$  and the angle of departure (AoD) $\theta^{AoD}$ are given by
\begin{equation}
    \label{equ72}
    \begin{aligned}
        \bm{\mathrm{a}}_r(\theta^{AoA}) = [1,e^{j\frac{2 \pi d }{\lambda}\sin{\theta^{AoA}}}, \cdots,e^{j\frac{2 \pi d }{\lambda}(D_r-1)\sin{\theta^{AoA}}} ]^{\mathrm{T}},\\
        \bm{\mathrm{a}}_t(\theta^{AoD}) = [1,e^{j\frac{2 \pi d }{\lambda}\sin{\theta^{AoD}}}, \cdots,e^{j\frac{2 \pi d }{\lambda}(D_t-1)\sin{\theta^{AoD}}} ]^{\mathrm{T}}.
    \end{aligned}
\end{equation}
In (\ref{equ72}), $D_r$ and $D_t$ represent the numbers of antennas at the receiver side and the transmitter side, respectively, and parameters $d$ and $\lambda$ denote the antenna spacing and wavelength, respectively. For simplicity, we set ${d}/{\lambda } = {1}/{2}$. The detailed settings of the above parameters are specified in the ``Communication and Radar model" block and the computing parameters $V_k$, $c_k$, $f_k^l$ and $f_{\mathrm{total}}^e$ are specified in the ``Computation model" block of Table \ref{table2}.
\begin{table}[t]
    \renewcommand\arraystretch{1.2}
    \centering
    \caption{Default Simulation Parameter Settings}
    \label{table2}
    \begin{tabular}{|p{2.5cm}|p{4.5cm}|}
        \hline
        Description &  Parameter and Value \\ \hline 
        \makecell[l]{Communication and\\ Radar model} & \makecell[l]{ $\mathrm{PL}_0 = 30$ \textrm{dB}, $d_0 = 1$ \textrm{m} \\ $\alpha_{\mathrm{bu}} = 3.75$, $\alpha_{\mathrm{r}} = 2.2$, $\alpha_{\mathrm{ru}} = 2.2$ \\ $\alpha_{\mathrm{uu}} = 2.2$, $B_\omega$ = $1$ \textrm{MHz} \\ $M = 4$, $N = 2$, $d = 2$ \\ $P_{\mathrm{t},k}^{\max} = 10$ \textrm{mW}, $\eta = 10$ \textrm{dB} \\$K_R = 3$, $L = 30$ \\ $\sigma_{\mathrm{c}} ^2 = 3.98 \times  10^{-12}$ \textrm{mW}\\ $\sigma_{\mathrm{s}} ^2= 3.98 \times  10^{-12}$ \textrm{mW}} \\ \hline
        Computation model    & \makecell[l]{ $V_k = [200, 300]$ \textrm{Kb} \\
        $c_k = [500, 600]$ \textrm{cycle/bit} \\ $f_k^l = [1\times10^8, 2\times10^8] $ \textrm{cycle/s} \\ $f_{\mathrm{total}}^e = 5 \times 10^{10} $ \textrm{cycle/s} } \\ \hline
        Weight & $\xi_k = 1/K$ \\ \hline
        Convergence criterion & $\epsilon = 10^{-3}$ \\ \hline
    \end{tabular}
\end{table}

In the subsequent subsections, we evaluate our proposed algorithms against the benchmark schemes listed below:
\begin{itemize}
    \item \textbf{RandPhase}: Assuming that the RIS phase shifts are randomly generated following a uniform distribution from 0 to $2\pi$, we jointly optimize the precoding and decoding matrices of UEs, the offloading volume, and the edge computational resource allocation. This optimization is conducted according to Algorithm \ref{alg:algorithm5} and Algorithm \ref{alg:algorithm4} in the single-UE and multi-UE scenarios, respectively, while excluding the design of the RIS phase shift.
    \item \textbf{Without RIS}: The RIS-UE links are assumed to be blocked. Only the precoding and decoding matrices of UEs, along with the offloading volume and the edge computational resource allocation, are jointly optimized using Algorithm \ref{alg:algorithm5} and Algorithm \ref{alg:algorithm4} in the single-UE and multi-UE scenarios, respectively.
\end{itemize}

\begin{figure*}
    \begin{center}
        \begin{minipage}{0.47\textwidth}
            \centering
            \subfloat[]{\includegraphics[width=1.6in]{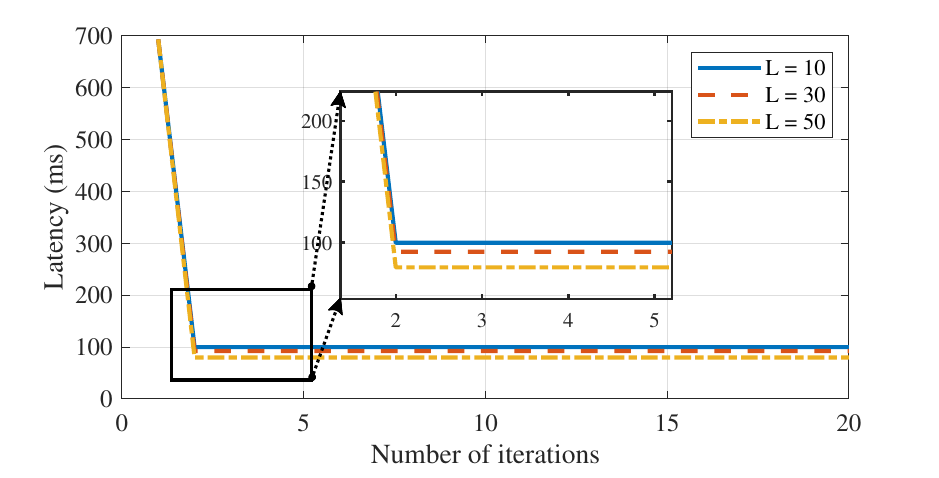} \label{figi1}} 
            \subfloat[]{ \includegraphics[width=1.6in]{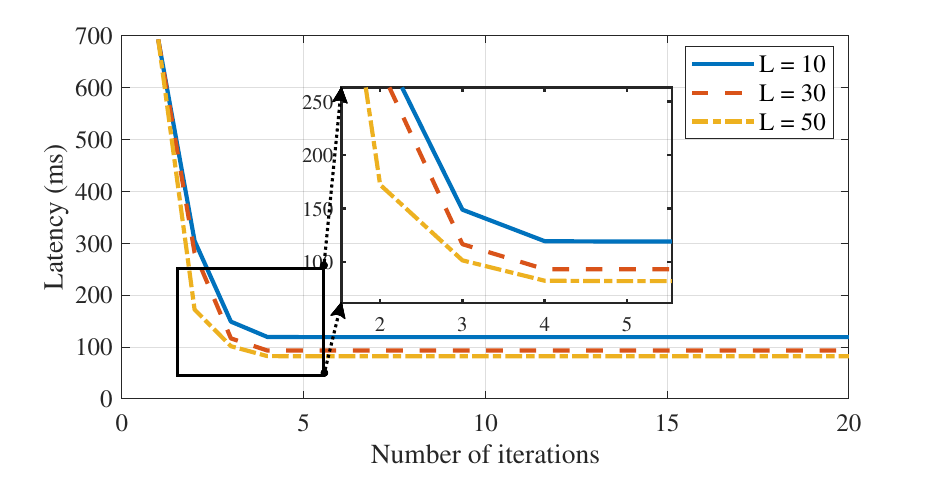} \label{figi2}} 
            \caption{Convergence behavior of the proposed algorithm for (a) the single-UE scenario using Algorithm \ref{alg:algorithm5} and (b) the multi-UE scenario using Algorithm \ref{alg:algorithm4}.}
        \label{figi}
        \end{minipage}
        \begin{minipage}{0.47\textwidth}
            \subfloat[]{\includegraphics[width=1.6in]{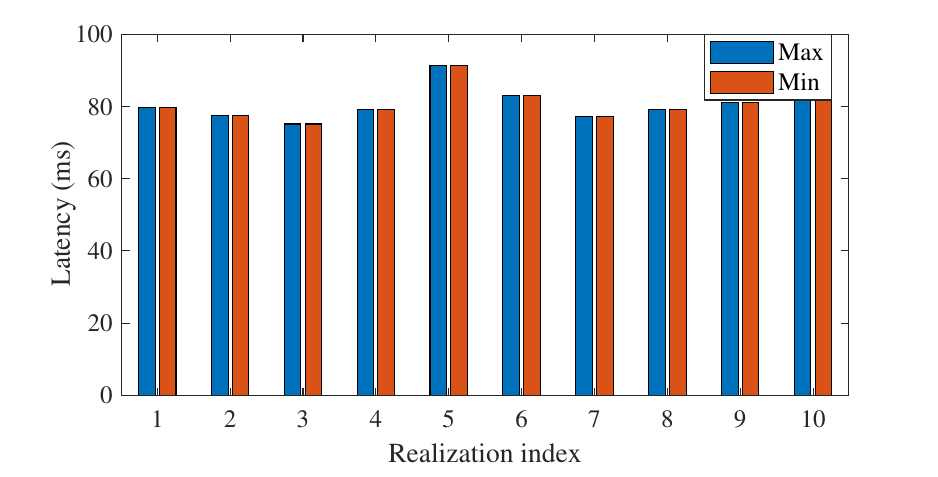} \label{figc1}} 
            \subfloat[]{\includegraphics[width=1.6in]{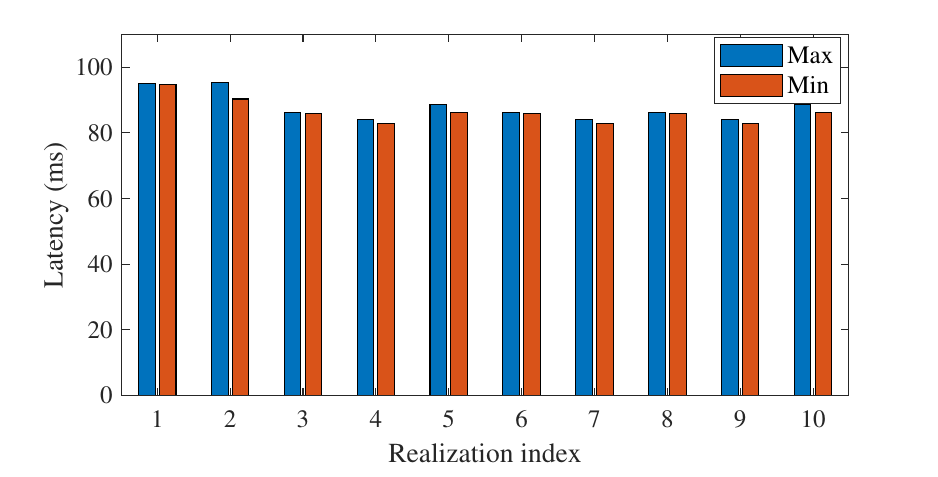} \label{figc2}} 
            \caption{Simulation results of the maximum and minimum latency versus the realization index obtained under 100 random initialization settings for (a) the single-UE scenario using Algorithm \ref{alg:algorithm5} and (b) the multi-UE scenario using Algorithm \ref{alg:algorithm4}.}
        \label{figc}
        \end{minipage}
    \end{center}
\end{figure*}

Fig. \ref{figi} presents the weighted latency versus the number of iterations for various numbers of RIS reflecting elements, i.e. $L = 10$, $L = 30$, and $L = 50$, in both single-UE and multi-UE scenarios. It can be seen from Fig. \ref{figi} that as the number of reflecting elements increases, the latency decreases due to the strengthened RIS-related link. Furthermore, both the proposed algorithms demonstrate rapid convergence within 5 iterations.

\hyphenation{diffe-rent}
\hyphenation{mini-mum}
Fig. \ref{figc} illustrates the algorithm performance under different initialization settings in both single-UE and multi-UE scenarios. We consider 10 channel realizations and parameters for random computational tasks, labeled 1\,-\,10. For each specific generation of channels and computational tasks, we randomly initialize the optimization variables 100 times and subsequently calculate the corresponding solutions using the algorithm outlined in the preceding section. Within these 100 result sets, the maximum latency, denoted by \textbf{Max}, represents the worst solution obtained from the algorithm, while the minimum latency, approximating the global optimal value, is designated as \textbf{Min}, as depicted in Fig. \ref{figc}. Notably, in the single-UE scenario, the values of \textbf{Max} and \textbf{Min} are similar, suggesting a high probability of achieving the globally optimal solution. In the multi-UE scenario, the maximum deviation is 5.6\%, indicating that the solutions generated by the proposed algorithm closely approach the optimal solution.

Fig. \ref{fign} depicts the weighted latency versus the number of RIS elements. In addition to the three aforementioned benchmark schemes, we also explore two scenarios where the continuous RIS phase shifts are quantized to 1 and 2 bits due to the practical hardware constraints. Specifically, in the \textbf{With RIS, 1 bit} scheme, the phase shift of RIS elements is limited to either 0 or $\pi$, whereas the \textbf{With RIS, 2 bit} scheme offers four options: $\{0, \frac{\pi}{2}, \pi, \frac{3\pi}{2}\}$. Our observations yield the following insights. Firstly, the latency of the \textbf{RandPhase} scheme exhibits a gradual decrease w.r.t. the number of RIS elements, widening the gap with the \textbf{Without RIS}  scheme. This suggests that integrating RIS can enhance communication performance even in the absence of intricate phase shift designs for the RIS elements.
Secondly, the performance improvement of the \textbf{With RIS} scheme over the \textbf{RandPhase} scheme is observed to amplify with the increase of RIS elements. This observation underscores the significance of meticulously designing the phase-shift response of RIS for enhancing beamforming gain with an increasing number of RIS elements.  Finally, in comparison to the \textbf{With RIS, 1 bit}  scheme, the latency in the \textbf{With RIS, 2 bit}  scheme is reduced and approaches the levels observed in the continuous phase shift scheme \textbf{With RIS}. This finding suggests that considering continuous phase design can serve as a tight bound for evaluating the practical discrete RIS phase shift configuration.
\begin{figure*}
    \begin{center}
        \begin{minipage}{0.3\textwidth}
            \subfloat[]{\includegraphics[width=2in]{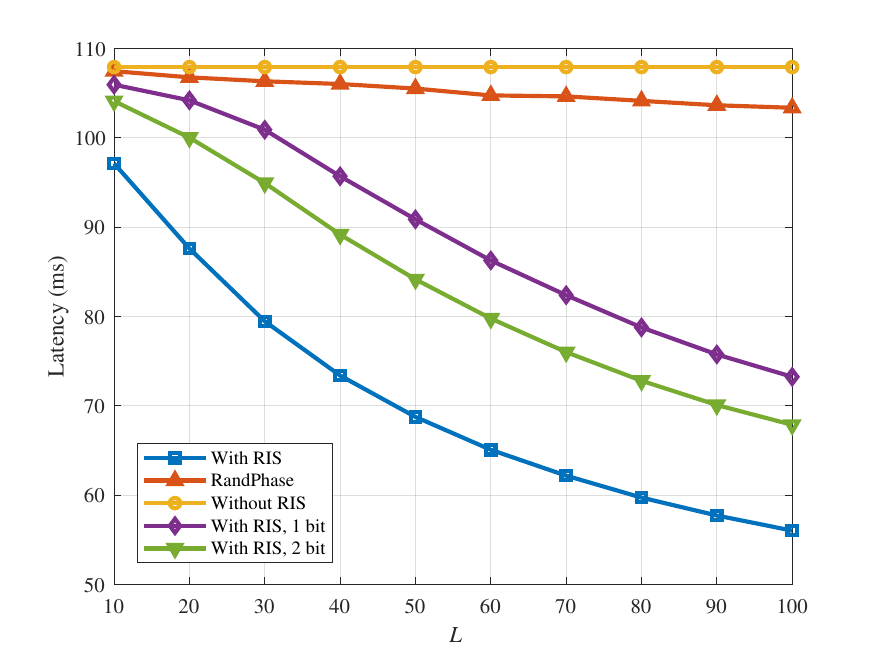} \label{fign1}}
            \\
            \subfloat[]{\includegraphics[width=2in]{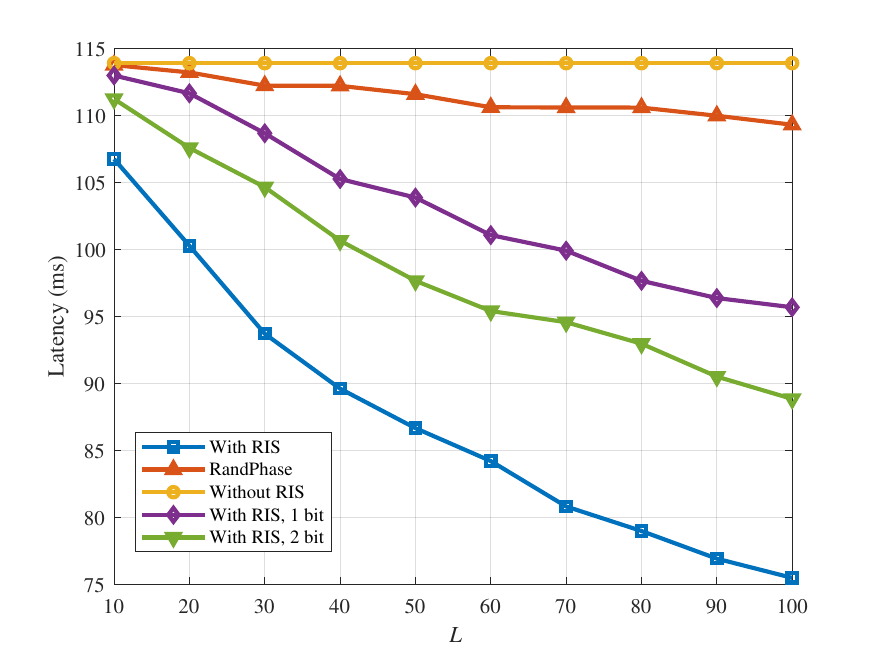} \label{fign2}}
            \caption{Latency versus the number of RIS elements for (a) the single-UE scenario using Algorithm \ref{alg:algorithm5} and (b) the multi-UE scenario using Algorithm \ref{alg:algorithm4}.}
            \label{fign}
        \end{minipage}
        \begin{minipage}{0.3\textwidth}
            \subfloat[]{\includegraphics[width=2in]{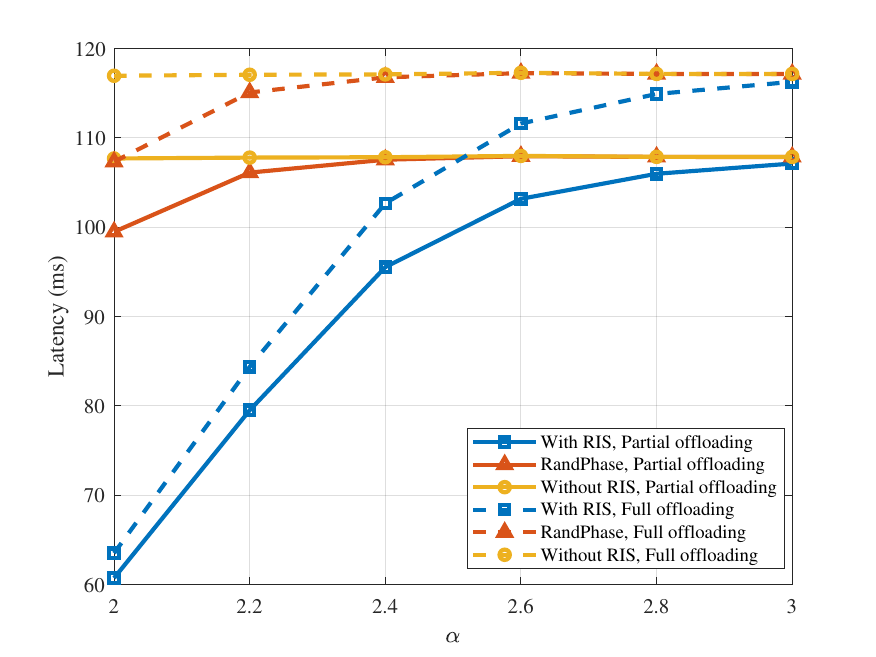} \label{figa1}}
            \\
            \subfloat[]{\includegraphics[width=2in]{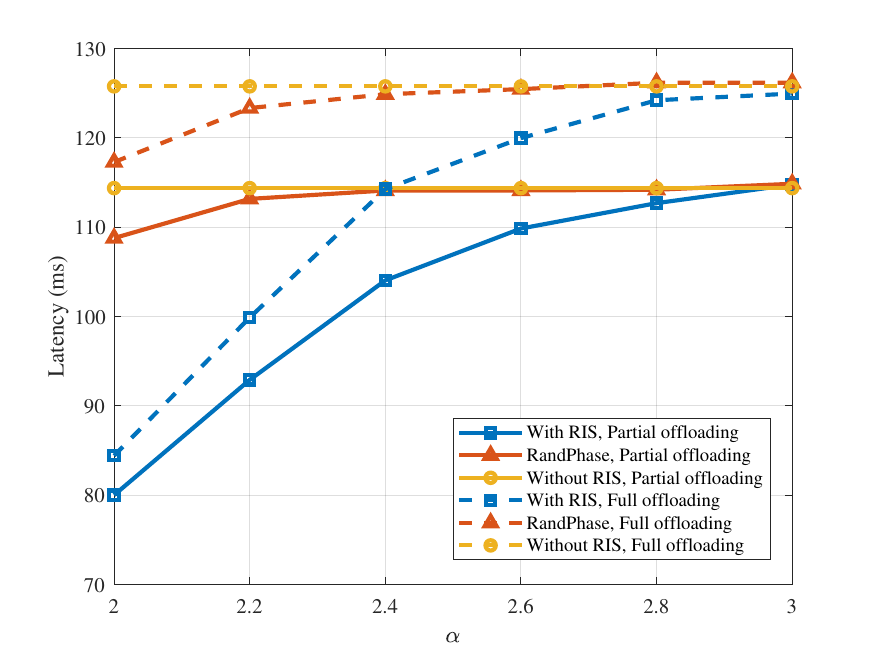} \label{figa2}}
            \caption{Latency versus the $\alpha _{RIS}$ path loss exponent for (a) the single-UE scenario using Algorithm \ref{alg:algorithm5} and (b) the multi-UE scenario using Algorithm \ref{alg:algorithm4}.}
            \label{figa}
        \end{minipage}
        \begin{minipage}{0.3\textwidth}
            \subfloat[]{\includegraphics[width=2in]{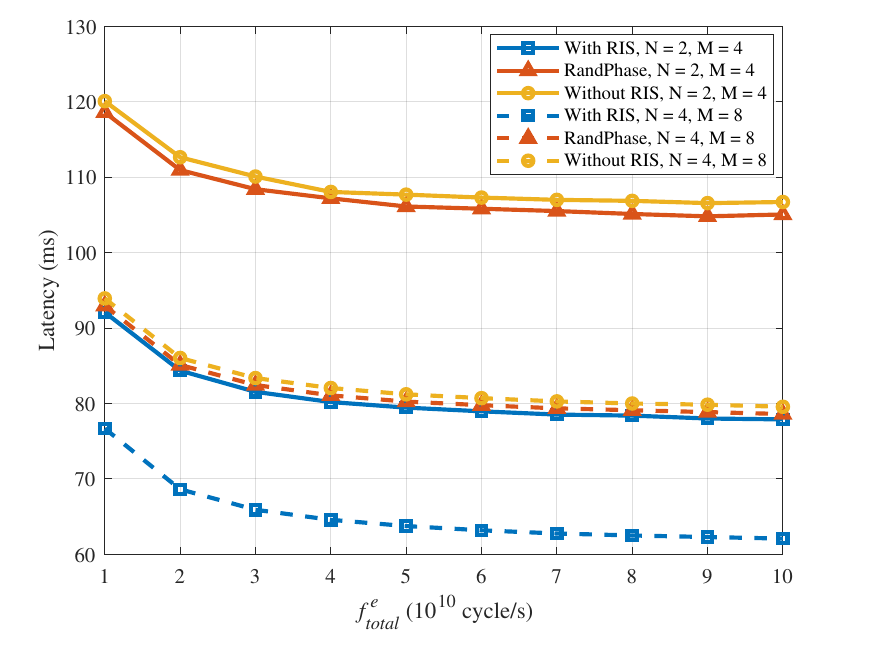} \label{figf1}}
            \\
            \subfloat[]{\includegraphics[width=2in]{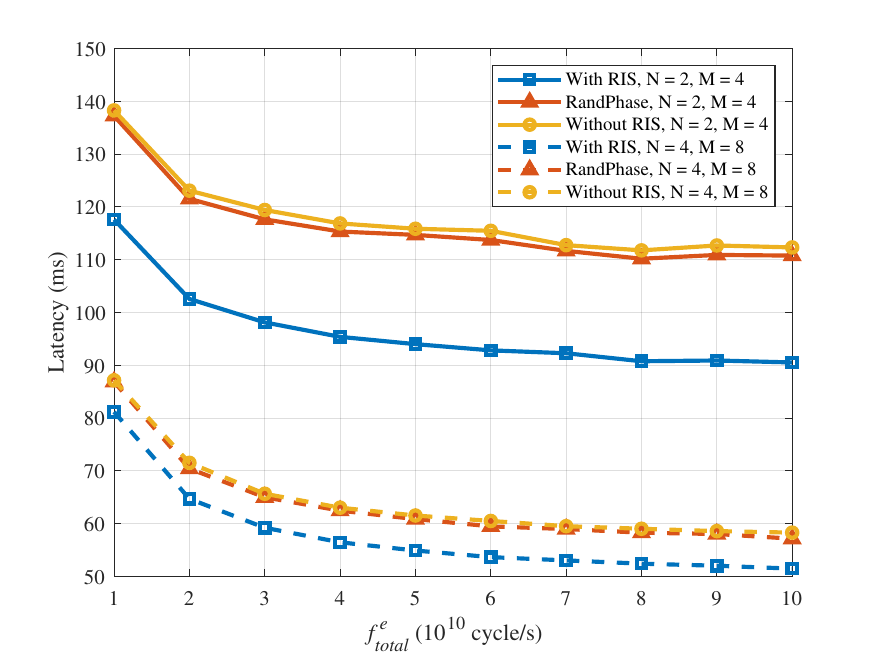} \label{figf2}}
            \caption{Latency versus the edge computing capability for (a) the single-UE scenario using Algorithm \ref{alg:algorithm5} and (b) the multi-UE scenario using Algorithm \ref{alg:algorithm4}.}
            \label{figf}
        \end{minipage}
    \end{center}
\end{figure*}

Fig. \ref{figa} plots the weighted latency versus the path loss exponent associated with RIS, specifically $\alpha_{\mathrm{ru}}$ and $\alpha_{\mathrm{r}}$, which are designated as  $\alpha_{RIS}$, under different offloading strategies. It can be seen from Fig. \ref{figa} that as expected, the latency escalates w.r.t. $\alpha_{RIS}$. In addition, the latency in the \textbf{RandPhase} scheme closely resembles that of the \textbf{Without RIS} scheme when the value of $\alpha_{RIS}$ is substantial. This is because the beamforming gain introduced by the RIS cannot compensate for an increasing channel path-loss $\alpha_{RIS}$ and the direct UE-BS links are dominated. This implies that in practical applications, strategic placement of the RIS is crucial to circumvent obstacles and achieve reduced values of $\alpha_{\mathrm{r}}$ and $\alpha_{\mathrm{ru}}$. Additionally, we examine the influence of offloading strategies. The latency achieved with the partial offloading strategy is smaller than that with the full offloading strategy, and the disparity between the two strategies is widened as $\alpha_{\mathrm{r}}$ increases.  The reason for this discrepancy can be explained as follows. In scenarios with favorable channel conditions, both strategies tend to offload the computational task to the edge server. However, as the RIS-related channel deteriorates, the partial offloading strategy becomes smarter in the processing part of the computational task locally.

\hyphenation{dece-lerates}
Fig. \ref{figf} shows the weighted latency versus the edge computing capability under different numbers of transmit and receive antennas. 
The latency initially decreases significantly as the processing rate of the ECS $f_{\mathrm{total}}^e$ rises, followed by a deceleration in the rate of latency reduction.
This is because at lower $f_{\mathrm{total}}^e$ values, the latency is primarily influenced by edge computing capability. By contrast, at higher $f_{\mathrm{total}}^e$ levels, the computation offloading delay is dominant in the overall latency. As a result, minimizing latency does not necessarily require equipping edge computing nodes with excessively powerful computational capabilities. Furthermore, a noteworthy decrease in latency occurs with an increase in the numbers of transmit and receive antennas, especially in the multi-UE scenario.  

Fig. \ref{figeta} portrays the weighted latency versus the SNR(SINR) threshold $\eta$. In the single-UE scenario, latency remains relatively constant despite variations in the threshold, whereas in multi-UE scenarios, latency increases with $\eta$. As demonstrated in Appendix \ref{app2}, the design of the optimal waveform essentially involves power allocation between the radar target and the BS. When the radar SNR requirement is low, there is no necessity to allocate power specifically for the radar target, as the signal transmitted to the BS inherently provides sufficient power to adequately meet the SNR requirement for the radar target. In a similar vein, within multi-UE scenarios, interference generated by other UEs necessitates the allocation of additional power to radar sensing in order to maintain the efficiency constraint ($\eta$), consequently resulting in an increase in latency.

\begin{figure}[h]
    \centering
    \includegraphics[width=2.3in]{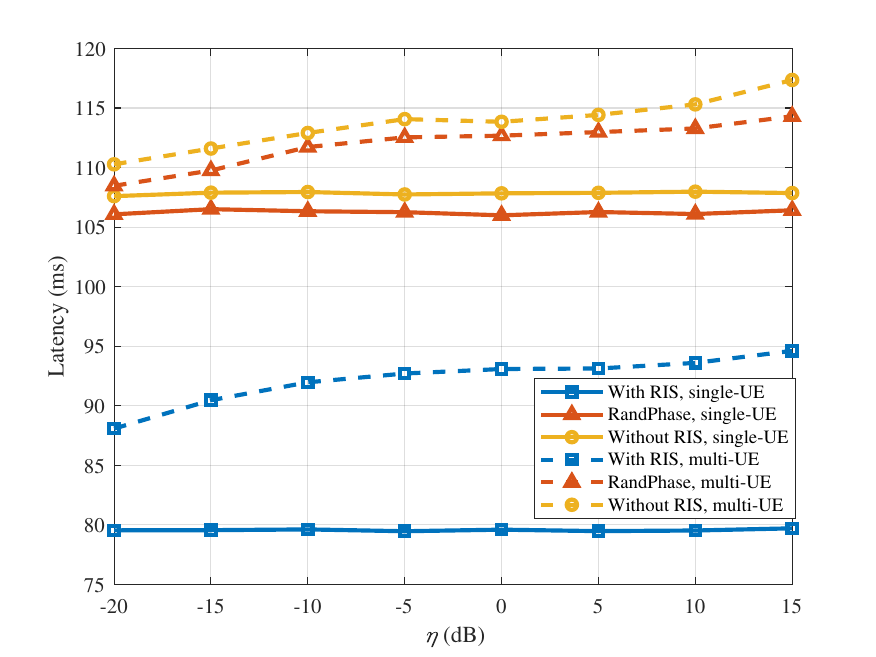}
    \caption{Latency versus the SNR threshold for the single-UE scenario using Algorithm \ref{alg:algorithm5} and the SINR threshold for the multi-UE scenario using Algorithm \ref{alg:algorithm4}.}
    \label{figeta}
\end{figure}
\section{Conclusion}
\hyphenation{signifi-cant}
In this paper, we investigated a framework for RIS-assisted ICSC systems with a User-centric focus, aimed at enhancing the performance of DFRC-enabled UEs. A latency minimization problem was formulated by jointly optimizing the precoding and decoding matrices of the uplink signal, the receive radar beamforming vector, the RIS phase shift vector, the offloading volume, and the edge computational resource allocation in both multi-UE scenario and single-UE scenario. In the multi-UE scenario, we proposed an algorithm based on the BCD method to address the non-convexity of the problem by optimizing the computational and beamforming settings alternately. Specifically, for the computational settings, a closed-form solution for the offloading volume was derived, and a low-complexity algorithm based on the bisection search method was applied to enhance the optimization of edge computing resource allocation. Then, two equivalent transformations were applied to the active and passive beamforming matrices to convert the OF, originally expressed as a non-convex sum-of-ratios, into a weighted sum-rate form, which was tackled by using a pair of efficient algorithms. Moreover, a low complexity algorithm with closed-form solutions was provided in the simplified single-UE scenario. Simulation results indicated a significant performance enhancement in latency with the proposed algorithm compared to the conventional approach operating without the integration of an RIS. Furthermore, numerical analysis confirmed the efficiency of our proposed algorithm. 
\numberwithin{equation}{section}
\begin{appendices}
\section{The Proof of Proposition \ref{pro1}}
\label{app1}
The partial Lagrange function of Problem $\mathcal{P}2\text{-}1$ is given by
\begin{equation}
    \mathcal{L} (\bm{\mathrm{F}}_{\mathrm{c}},\bm{\mathrm{W}}_{\mathrm{c}},\bm{\mathrm{\theta}},\bm{\mathrm{\lambda }}) = \sum _{k=1}^K \lambda _k + \sum _{k=1}^K \delta _k [\xi _k v _k - \lambda _k R_k(\bm{\mathrm{F}}_{\mathrm{c}},\bm{\mathrm{W}}_{\mathrm{c}},\bm{\mathrm{\theta}})],
\end{equation}
where $\{ \delta _k\}  $ is the Lagrange multiplier. Given that $\{\bm{\mathrm{F}}_{\mathrm{c}}^{\star},\bm{\mathrm{W}}_{\mathrm{c}}^{\star},\bm{\mathrm{w}}_{\mathrm{s}}^{\star },\bm{\theta }^{\star},\bm{\mathrm{\lambda }}^{\star}\}$ represents the solution to Problem $\mathcal{P}2\text{-}1$, there exists $\bm{\mathrm{\delta}}^{\star}$ that fulfills the following KKT conditions:
\begin{align}
    & \frac{\partial \mathcal{L}}{\partial \theta _k} = -\delta _k^{\star} \lambda _k^{\star}\bigtriangledown R_k(\bm{\mathrm{F}}_{\mathrm{c}}^{\star},\bm{\mathrm{W}}_{\mathrm{c}}^{\star},\bm{\mathrm{\theta}}^{\star}) = 0,  \forall k \in \mathcal{K},\label{equapp3}\\ 
    & \frac{\partial \mathcal{L}}{\partial \mathrm{F}_{c,k}} =  -\delta _k^{\star} \lambda _k^{\star}\bigtriangledown R_k(\bm{\mathrm{F}}_{\mathrm{c}}^{\star},\bm{\mathrm{W}}_{\mathrm{c}}^{\star},\bm{\mathrm{\theta}}^{\star}) = 0, \forall k \in \mathcal{K},\label{equapp4}\\ 
    & \frac{\partial \mathcal{L}}{\partial \mathrm{W}_{c,k}} =  -\delta _k^{\star} \lambda _k^{\star}\bigtriangledown R_k(\bm{\mathrm{F}}_{\mathrm{c}}^{\star},\bm{\mathrm{W}}_{\mathrm{c}}^{\star},\bm{\mathrm{\theta}}^{\star}) = 0, \forall k \in \mathcal{K},\label{equapp5}\\ 
    & \frac{\partial \mathcal{L}}{\partial \mathrm{\lambda }_{ k}} = 1-\delta _k^{\star} R_k(\bm{\mathrm{F}}_{\mathrm{c}}^{\star},\bm{\mathrm{W}}_{\mathrm{c}}^{\star},\bm{\mathrm{\theta}}^{\star}) = 0, \forall k \in \mathcal{K},\label{equapp1}\\ 
    & \delta _k^{\star} [\xi _k v _k - \lambda _k^{\star} R_k(\bm{\mathrm{F}}_{\mathrm{c}}^{\star},\bm{\mathrm{W}}_{\mathrm{c}}^{\star},\bm{\mathrm{\theta}}^{\star})] = 0, \forall k \in \mathcal{K},\label{equapp2}\\ 
    & \delta _k^{\star} \geqslant 0, \forall k \in \mathcal{K},\label{equapp6}\\ 
    & \xi _k v _k - \lambda _k^{\star} R_k(\bm{\mathrm{F}}_{\mathrm{c}}^{\star},\bm{\mathrm{W}}_{\mathrm{c}}^{\star},\bm{\mathrm{\theta}}^{\star}) \leqslant 0, \forall k \in \mathcal{K},\label{equapp7}\\ 
    & 0 \leqslant \theta _k^{\star} \leqslant 2 \pi, \forall k \in \mathcal{K}. \label{equapp8}
\end{align}
Since $R_k(\bm{\mathrm{F}}_{\mathrm{c}},\bm{\mathrm{\theta}}) > 0$, (\ref{equapp1}) is equivalent to 
\begin{equation}
    \delta _k^{\star} = \frac{1}{R_k(\bm{\mathrm{F}}_{\mathrm{c}}^{\star},\bm{\mathrm{W}}_{\mathrm{c}}^{\star},\bm{\mathrm{\theta}}^{\star}) }, \forall k \in \mathcal{K}.
\end{equation}
From (\ref{equapp2}), we can derive that
\begin{equation}
    \lambda _k ^{\star} = \frac{\xi _k v _k}{R_k(\bm{\mathrm{F}}_{\mathrm{c}}^{\star},\bm{\mathrm{W}}_{\mathrm{c}}^{\star},\bm{\mathrm{\theta}}^{\star})},  \forall k \in \mathcal{K}.
\end{equation}
When we set $\bm{\mathrm{\lambda }} = \bm{\mathrm{\lambda }}^{\star}$ and $\bm{\mathrm{\delta}} = \bm{\mathrm{\delta}}^{\star}$, the KKT conditions of Problem $\mathcal{P}2\text{-}2$ are the same as (\ref{equapp3}), (\ref{equapp4}), (\ref{equapp5}), and (\ref{equapp8}). Thus, Proposition \ref{pro1} is proved.
\section{The Proof of Proposition \ref{pro2}}
\label{app2}
First, we demonstrate that when the OF of Problem $\mathcal{P}3\text{-}3$ is maximized, the constraint in (\ref{transmit}) achieves the upper bound $\bm{\mathrm{f}}_{\mathrm{c}} \bm{\mathrm{f}}_{\mathrm{c}}^{\mathrm{H}} = P_{t}$ by using the proof by contradiction. Assuming that $\bm{\mathrm{f}}_{\mathrm{c}}^{\star}$ represents the optimal solution of Problem $\mathcal{P}3\text{-}3$ such that $\bm{\mathrm{f}}_{\mathrm{c}}^{\star} {\bm{\mathrm{f}}_{\mathrm{c}}^{\star \mathrm{H}}} < P_t$, we can derive $\hat{\bm{\mathrm{f}}_{\mathrm{c}}} = (1+\beta)\bm{\mathrm{f}}_{\mathrm{c}}^{\star}$ satisfying $\hat{\bm{\mathrm{f}}_{\mathrm{c}}} \hat{\bm{\mathrm{f}}_{\mathrm{c}}}^{\mathrm{H}} = P_{t}$, where $\beta \geqslant 0$, as a feasible solution of Problem $\mathcal{P}3\text{-}3$. 
Since the OF of Problem $\mathcal{P}3\text{-}3$ when $\bm{\mathrm{f}}_{\mathrm{c}} = \hat{\bm{\mathrm{f}}_{\mathrm{c}}}$ surpasses that when $\bm{\mathrm{f}}_{\mathrm{c}} =\bm{\mathrm{f}}_{\mathrm{c}}^{\star}$, the OF of Problem $\mathcal{P}3\text{-}3$ is maximized if and only if when $\bm{\mathrm{f}}_{\mathrm{c}} \bm{\mathrm{f}}_{\mathrm{c}}^{\mathrm{H}} = P_{t}$. Thus, Problem $\mathcal{P}3\text{-}3$ can be rewritten as
\begin{subequations}
\label{equfff}
    \begin{align}
    \max_{\bm{\mathrm{f}}_{\mathrm{c}}} \quad &
        \left\lvert \bm{\mathrm{h}}^{\mathrm{H}} \bm{\mathrm{f}}_{c} \right\rvert ^2 \\
        s.t. \quad &
         \left\lvert \bm{\mathrm{f}}_{c} \right\rvert ^2 = P_{t},\\
        & \left\lvert \bm{\mathrm{g}}^{\mathrm{H}} \bm{\mathrm{f}}_{c} \right\rvert ^2 \geqslant  \eta. \label{equb2c}
    \end{align}
\end{subequations}
If $\eta > P_{t} \left\lVert \bm{\mathrm{g}} \right\rVert ^2$, by applying the Cauchy's inequality, we have $|\bm{\mathrm{g}}^{\mathrm{H}} \bm{\mathrm{f}}_{c}|^2 \leqslant  \left\lVert \bm{\mathrm{g}} \right\rVert ^2 \left\lVert \bm{\mathrm{f}}_{c} \right\rVert ^2 = P_{t} \left\lVert \bm{\mathrm{g}} \right\rVert ^2 < \eta$, which is in conflict with constraint (\ref{equb2c}). Thus, Problem (\ref{equfff}) is feasible only when $ \eta \leqslant P_{t} \left\lVert \bm{\mathrm{g}} \right\rVert ^2$. Case \ref{case3} in Proposition \ref{pro2} is proved.

As Problem (\ref{equfff}) is convex and satisfies Slater's condition \cite{44}, the KKT conditions are utilized to address this problem. The Lagrange function of Problem (\ref{equfff}) is expressed as
\begin{align}
    \mathcal{L}(\bm{\mathrm{f}}_{\mathrm{c}}) = -\bm{\mathrm{h}}^{\mathrm{H}} \bm{\mathrm{f}}_{\mathrm{c}} \bm{\mathrm{f}}_{\mathrm{c}}^{\mathrm{H}} \bm{\mathrm{h}} + \varphi 
    (\eta - \bm{\mathrm{g}}^{\mathrm{H}} \bm{\mathrm{f}}_{\mathrm{c}} \bm{\mathrm{f}}_{\mathrm{c}}^{\mathrm{H}} \bm{\mathrm{g}}) + \iota 
    (\bm{\mathrm{f}}_{\mathrm{c}} \bm{\mathrm{f}}_{\mathrm{c}}^{\mathrm{H}} - P_{t} ),
\end{align}
where $\varphi$ and $\iota $ are the Lagrange multipliers. 
Then, the KKT conditions for (\ref{equfff}) can be formulated as follows:
\begin{align}
& \frac{\partial \mathcal{L}}{\partial \bm{\mathrm{f}}_{\mathrm{c}}^{\star}} = - \bm{\mathrm{h}}^{\mathrm{H}} \bm{\mathrm{f}}_{\mathrm{c}} \bm{\mathrm{h}} - \varphi \bm{\mathrm{g}}^{\mathrm{H}} \bm{\mathrm{f}}_{\mathrm{c}}\bm{\mathrm{g}} + \iota \bm{\mathrm{f}}_{\mathrm{c}} = \bm{0}, \label{equ83} \\ 
& \eta - \bm{\mathrm{g}}^{\mathrm{H}} \bm{\mathrm{f}}_{\mathrm{c}} \bm{\mathrm{f}}_{\mathrm{c}}^{\mathrm{H}} \bm{\mathrm{g}} \leqslant 0,  \label{equ84} \\
& \bm{\mathrm{f}}_{\mathrm{c}} \bm{\mathrm{f}}_{\mathrm{c}}^{\mathrm{H}} - P_{t} = 0,  \label{equ85} \\
& \varphi \geqslant 0,  \label{equ86} \\
& \varphi 
    (\eta - \bm{\mathrm{g}}^{\mathrm{H}} \bm{\mathrm{f}}_{\mathrm{c}} \bm{\mathrm{f}}_{\mathrm{c}}^{\mathrm{H}} \bm{\mathrm{g}}) = 0 \label{equ87} .
\end{align}
  (\ref{equ83}) is equivalent to 
\begin{equation}
    \label{equ89}
    \bm{\mathrm{f}}_{\mathrm{c}} = \frac{\bm{\mathrm{h}}^{\mathrm{H}} \bm{\mathrm{f}}_{\mathrm{c}}}{\iota } \bm{\mathrm{h}} + \frac{\varphi \bm{\mathrm{g}}^{\mathrm{H}} \bm{\mathrm{f}}_{\mathrm{c}}}{\iota } \bm{\mathrm{g}}, 
\end{equation}
which indicates $\bm{\mathrm{f}}_{\mathrm{c}} \in \mathrm{span} (\bm{\mathrm{h}}, \bm{\mathrm{g}})$. For convenience, $\bm{\mathrm{f}}_{\mathrm{c}}$ can be rewritten as
\begin{equation}
    \label{equ90}
    \bm{\mathrm{f}}_{\mathrm{c}} = a \bm{\mathrm{h}} + b \bm{\mathrm{g}},
\end{equation}
where the coefficients $a$ and $b$ denote the amplitudes of the power budget allocated to communication and radar sensing, respectively. By substituting (\ref{equ90}) into (\ref{equ85}), we obtain that
\begin{align}
    |a|^2 \bm{\mathrm{h}}^{\mathrm{H}} \bm{\mathrm{h}} + a^{\star} b \bm{\mathrm{h}}^{\mathrm{H}} \bm{\mathrm{g}} + ab^{\star} \bm{\mathrm{g}}^{\mathrm{H}}\bm{\mathrm{h}} +  |b|^2 \bm{\mathrm{g}}^{\mathrm{H}} \bm{\mathrm{g}} = P_t. \label{equ93} 
\end{align}
Moreover, (\ref{equ87}) suggests that either $\varphi$ or $\eta - \bm{\mathrm{g}}^{\mathrm{H}} \bm{\mathrm{f}}_{\mathrm{c}} \bm{\mathrm{f}}_{\mathrm{c}}^{\mathrm{H}} \bm{\mathrm{g}}$ equals to 0. Next, we discuss the different cases separately.

\begin{enumerate}
    \item When $\varphi = 0$, (\ref{equ89}) is transformed into 
\end{enumerate}
\begin{equation}
\label{equ91}
    \bm{\mathrm{f}}_{\mathrm{c}} = \frac{\bm{\mathrm{h}}^{\mathrm{H}} \bm{\mathrm{f}}_{\mathrm{c}}}{\iota } \bm{\mathrm{h}}. 
\end{equation}
By substituting (\ref{equ90}) into (\ref{equ91}), we derive that
\begin{align}
    a \bm{\mathrm{h}} + b \bm{\mathrm{g}} = 
    \frac{a \bm{\mathrm{h}}^{\mathrm{H}} \bm{\mathrm{h}} + b\bm{\mathrm{h}}^{\mathrm{H}} \bm{\mathrm{g}} }{\iota } \bm{\mathrm{h}}. \label{equ94}
\end{align}
From (\ref{equ94}), it can be deduced that $b = 0$. By substituting $b = 0$ into   (\ref{equ93}), the value of $a$ can be determined as
\begin{equation}
    |a| = \frac{\sqrt{P_t}}{||\bm{\mathrm{h}}||},
\end{equation}
where the phase of $a$ is arbitrary. Meanwhile, it can be derived from (\ref{equ84}) that 
\begin{equation}
    \eta \leqslant \frac{P_t |\bm{\mathrm{g}}^{\mathrm{H}} \bm{\mathrm{h}}|^2 }{||\bm{\mathrm{h}}|| ^2}.
\end{equation}
Thus, Case \ref{case1} in Proposition \ref{pro2} is proved.

\begin{enumerate}\setcounter{enumi}{1}
    \item When $\eta - \bm{\mathrm{g}}^{\mathrm{H}} \bm{\mathrm{f}}_{\mathrm{c}} \bm{\mathrm{f}}_{\mathrm{c}}^{\mathrm{H}} \bm{\mathrm{g}} = 0$, by substituting   (\ref{equ90}) into $\eta - $
\end{enumerate}
$\bm{\mathrm{g}}^{\mathrm{H}} \bm{\mathrm{f}}_{\mathrm{c}} \bm{\mathrm{f}}_{\mathrm{c}}^{\mathrm{H}} \bm{\mathrm{g}} = 0$ and (\ref{equ89}), we have
\begin{align}
        & |a|^2 \bm{\mathrm{g}}^{\mathrm{H}} \bm{\mathrm{h}} \bm{\mathrm{h}}^{\mathrm{H}} \bm{\mathrm{g}} + a^{\star}b \bm{\mathrm{g}}^{\mathrm{H}} \bm{\mathrm{g}} \bm{\mathrm{h}}^{\mathrm{H}} \bm{\mathrm{g}} + ab^{\star} \bm{\mathrm{g}}^{\mathrm{H}} \bm{\mathrm{h}} \bm{\mathrm{g}}^{\mathrm{H}} \bm{\mathrm{g}}  \nonumber \\ 
        &+|b|^2 \bm{\mathrm{g}}^{\mathrm{H}} \bm{\mathrm{g}} \bm{\mathrm{g}}^{\mathrm{H}} \bm{\mathrm{g}} = \eta, \label{equ97}  \\ 
        &a \bm{\mathrm{h}} + b \bm{\mathrm{g}} = \frac{a \bm{\mathrm{h}}^{\mathrm{H}} \bm{\mathrm{h}} + b\bm{\mathrm{h}}^{\mathrm{H}} \bm{\mathrm{g}} }{\iota } \bm{\mathrm{h}} + \frac{\varphi (a \bm{\mathrm{g}}^{\mathrm{H}} \bm{\mathrm{h}} + b\bm{\mathrm{g}}^{\mathrm{H}} \bm{\mathrm{g}})}{\iota} \bm{\mathrm{g}} \label{equ98}.
\end{align}
From (\ref{equ93}) and (\ref{equ97}), it can be derived that
\begin{equation}
\label{a}
    |a| = \sqrt{\frac{P_t ||\bm{\mathrm{g}}||^2 - \eta}{||\bm{\mathrm{h}}||^2 ||\bm{\mathrm{g}}||^2 - |\bm{\mathrm{h}}^{\mathrm{H}} \bm{\mathrm{g}}|^2 }}.
\end{equation}
From (\ref{equ98}), we obtain that
\begin{align}
    &\angle(a) = \angle(a \bm{\mathrm{h}}^{\mathrm{H}} \bm{\mathrm{h}} + b\bm{\mathrm{h}}^{\mathrm{H}} \bm{\mathrm{g}}) = \angle(b\bm{\mathrm{h}}^{\mathrm{H}} \bm{\mathrm{g}}) , \label{equb19} \\
    &\angle(b) = \angle(a \bm{\mathrm{g}}^{\mathrm{H}} \bm{\mathrm{h}} + b\bm{\mathrm{g}}^{\mathrm{H}} \bm{\mathrm{g}}) = \angle(a \bm{\mathrm{g}}^{\mathrm{H}} \bm{\mathrm{h}}).\label{equb20}
\end{align}
It follows
\begin{equation}
\label{ab}
    \angle(a) - \angle(b) = \angle(\bm{\mathrm{h}}^{\mathrm{H}} \bm{\mathrm{g}}).
\end{equation}
Meanwhile, by substituting   (\ref{a}) and (\ref{ab}) into   (\ref{equ93}), we derive that 
\begin{equation}
    |b| = \frac{\sqrt{\eta}}{||\bm{\mathrm{g}}||^2} - 
    \frac{|\bm{\mathrm{h}}^{\mathrm{H}} \bm{\mathrm{g}} |}{||\bm{\mathrm{g}}||^2} |a|.
\end{equation}
Furthermore, since $|b| \geqslant 0$, the parameter $\eta$ satisfies $\eta \geqslant  \frac{P_t |\bm{\mathrm{g}}^{\mathrm{H}} \bm{\mathrm{h}}|^2 }{|\bm{\mathrm{h}}| ^2}$. Thus, Case \ref{case2} in Proposition \ref{pro2} is proved.

\end{appendices}
\bibliographystyle{IEEEtran} 
\bibliography{IEEEabrv,ICSC.bib}
\end{document}